\documentclass[format=acmsmall, review=false, screen=true]{acmart}

\usepackage{booktabs} % For formal tables

\makeatletter
\newcommand*{\rom}[1]{\expandafter\@slowromancap\romannumeral #1@}
\makeatother

\usepackage{algpseudocode,algorithm,algorithmicx}

\algrenewcommand\algorithmicrequire{\textbf{Precondition:}}
\algrenewcommand\algorithmicensure{\textbf{Postcondition:}}
\usepackage{xcolor}

\definecolor{cornellred}{rgb}{0.7, 0.11, 0.11}
\usepackage{hyperref}
\hypersetup{
    colorlinks = true,
    linkcolor=cornellred,
    citecolor=blue,
    linkbordercolor = {white}
}

\theoremstyle{plain}

\newenvironment{remark}[1][Remark]{\begin{trivlist}
\item[\hskip \labelsep {\bfseries #1}]}{\end{trivlist}}

\usepackage{accents}

\newcommand{\SSR}{\Gamma}
\newcommand{\rbf}{\mathbf{r}}
\newcommand{\sbf}{\mathbf{s}}

\newcommand{\xbf}{\mathbf{x}}
\newcommand{\vbf}{\mathbf{v}}
\newcommand{\alphabf}{\pmb{\alpha}}
\newcommand{\lagrange}{\mathcal{L}}
\newcommand{\opt}{\text{OPT}}
\newcommand{\alg}[1]{\text{ALG}_{#1}}
\newcommand{\hist}[1]{\mathcal{H}_{#1}}

\newcommand{\nBidder}{n}
\newcommand{\type}[1]{t^{#1}}
\newcommand{\typeS}[1]{\mathcal{T}^{#1}}
\newcommand{\bid}[1]{s^{#1}}
\newcommand{\typeP}{\mathbf{t}}
\newcommand{\bidP}{\mathbf{s}}
\newcommand{\outS}{\mathcal{O}}
\newcommand{\prior}[1]{F^{#1}}
\newcommand{\priorJ}{\mathbf{F}}
\newcommand{\alloc}{\mathcal{A}}
\newcommand{\mech}{\mathcal{M}}
\newcommand{\payP}{\mathbf{p}}
\newcommand{\pay}[1]{p^{#1}}
\newcommand{\val}{\textrm{val}}
\newcommand{\valk}[1]{\textrm{val}^{#1}}

\newcommand{\prob}[2][]{\text{\bf Pr}\ifthenelse{\not\equal{}{#1}}{_{#1}}{}\!\left[#2\right]}
\newcommand{\expect}[2][]{\text{\bf E}\ifthenelse{\not\equal{}{#1}}{_{#1}}{}\!\left[#2\right]}
\newcommand{\given}{\,\mid\,}

\newcommand{\ex}[2]{\expect[#1]{#2}}
\newcommand{\Ex}[2]{\expect[#1]{#2}}

\newcommand{\typeO}[1]{\mathbf{t}^{#1}}

\newcommand{\bern}[1]{\text{Bern}[#1]}

\newcommand{\Di}[1]{\mathcal{D}_{#1}}
\newcommand{\vmax}{v_{\max}}

\DeclareMathOperator{\argmax}{argmax}
\DeclareMathOperator{\argmin}{argmin}
\DeclareMathOperator{\OPT}{OPT}

\PassOptionsToPackage{authoryear,round}{natbib}
\acmJournal{JACM}
\acmYear{2020}
\setcopyright{acmlicensed}
\acmDOI{}

\received{May 2018}
\received[revised]{June 2020}
\received[accepted]{Nov 2020}
\begin{document}
% Title portion. Note the short title for running heads
\title[Bernoulli Factories and Black-Box Reductions]{Bernoulli Factories and Black-Box Reductions\\ in Mechanism Design}

\author{Shaddin Dughmi}
\affiliation{%
  \institution{Department of Computer Science, University of Southern California}
   \city{Los Angeles}
  \state{California}
  \country{USA}
  }
\email{shaddin@usc.edu}

\author{Jason Hartline}
\affiliation{%
  \institution{Computer Science Department, Northwestern University}
  \city{Evanston}
  \state{Illinois}
  \country{USA}
}
\email{hartline@eecs.northwestern.edu}

\author{Robert D. Kleinberg}
\affiliation{%
  \institution{Department of Computer Science, Cornell University}
   \city{Ithaca}
  \state{New York}
  \country{USA}
}
\email{rdk@cs.cornell.edi}

\author{Rad Niazadeh}
\affiliation{%
  \institution{Chicago Booth School of Business, University of Chicago}
  \city{Chicago}
  \state{Illinois}
  \country{USA}
}
\email{rad.niazadeh@chicagobooth.edu}

\begin{abstract}
We provide a polynomial time reduction from Bayesian
incentive compatible mechanism design to Bayesian algorithm design for welfare maximization problems. Unlike prior results, our reduction achieves exact incentive compatibility for problems with multi-dimensional and continuous type spaces. The key technical barrier preventing exact incentive compatibility in prior black-box reductions is that repairing violations of incentive constraints requires understanding the distribution of the mechanism's output, which is typically \#P-hard to compute. Reductions that instead estimate the output distribution by sampling inevitably  suffer from sampling error, which typically precludes exact
incentive compatibility. We overcome this barrier by employing and generalizing the
computational model in the literature on \emph{Bernoulli Factories}. In a Bernoulli factory problem, one is given a function mapping the bias of an ``input coin'' to that of an ``output coin'', and the challenge is to efficiently simulate the output coin given only sample access to the input coin. This is the key ingredient in designing an incentive compatible mechanism for bipartite matching,
which can be used to make the approximately incentive compatible reduction of \citet{HKM-15} exactly incentive compatible.

\end{abstract}





\maketitle

\renewcommand{\shortauthors}{Dughmi et al.}

\section{Introduction}
%% our result, and comparison to the literature.
We resolve an open question from \citet{HKM-11,HKM-15}, which is considered as one of the fundamental algorithmic questions in the Bayesian mechanism design:
{\em There is a polynomial time reduction from Bayesian incentive
  compatible mechanism design to Bayesian algorithm design for welfare
  maximization problems.} A Bayesian algorithm is one that
  performs well in expectation when the input is drawn from a known
  distribution. By polynomial time, we mean polynomial in the number
  of agents and the combined ``size'' of their type spaces. The key
distinction between our result and those of \citet{HKM-11,HKM-15} is
%that it satisfies \emph{exact} Bayesian incentive compatibility even
%for problems in which agents' preferences are multidimensional and
%drawn from a continuous space.
that both (a) the agents' preferences can be
multi-dimensional and from a continuous space (rather than
single-dimensional or from a discrete space), and (b) the
resulting mechanism is exactly Bayesian incentive compatible (rather
than approximately Bayesian incentive compatible).

%% what is IC.
%% blackbox reductions
%% related to  in range.
A mechanism solicits preferences from agents, i.e., how much each
agent prefers each outcome, and then chooses an outcome.  {\em
  Incentive compatibility} of a mechanism requires that, though agents
could misreport their preferences, it is not in any agent's best
interest to do so.  A quintessential research problem at the
intersection of mechanism deign and approximation algorithms is to
identify black-box reductions from approximation mechanism design to
approximation algorithm design.  The key algorithmic property that
makes a mechanism incentive compatible is that, from any individual agent's
perspective, it must be {\em maximal-in-range}, specifically, the
outcome selected maximizes the agent's utility less some cost that is
a function of the outcome (e.g., this cost function can depend on
other agents' reported preferences).\footnote{In general, one can think of maximal-in-range for all agents, meaning that the outcome selected maximizes agents' social welfare less some cost that is a function of the outcome among a particular range of feasible outcomes.}
% (more precisely, they are {\em affine
%  maximizers}), meaning that they always select a utility maximizing
%outcome for the agent from some menu of outcomes (the ``range''). 
%This
%menu is determined by other agents' reported preferences, though must
%be independent of the report of the agent in question.

%
% how do Bayesian black box reductions work
%
The black-box reductions from Bayesian mechanism design to Bayesian
algorithm design in the literature are based on obtaining an
understanding of the distribution of outcomes produced by the
algorithm through simulating the algorithm on samples from agents'
preferences. 
%In prior work, reductions achieving exact incentive
%compatibility have been restricted to settings in which the
%distribution of outcomes can be characterized exactly. 
Notice that, even for structurally simple problems, calculating the
exact probability that a given outcome is selected by an algorithm can be
\#P-hard.  For example, \citet{HKM-15} show such a result for
calculating the probability that a matching in a bipartite graph is
optimal, for a simple explicitly given distribution of edge weights.
A black-box reduction for mechanism design must therefore produce
exactly maximal-in-range outcomes merely from samples.  This challenge motivates new questions for algorithm design from samples.

\paragraph{The Expectations from Samples Model.}  In traditional
algorithm design, the inputs are specified to the algorithm
exactly. In this paper, we formulate the {\em expectations from
  samples} model.  This model calls for drawing an outcome from a
distribution that is a precise function of the expectations of some random
sources that are given only by sample access.  Formally, a problem for
this model is described by a function $f: [0,1]^n \to \Delta(X)$ where
$X$ is an abstract set of feasible outcomes and $\Delta(X)$ is the
family of probability distributions over $X$.  For any $n$ input
distributions on support $[0,1]$ with unknown expectations
$\boldsymbol \mu = (\mu_1,\ldots,\mu_n)$, an algorithm for such a
problem, with only sample access to each of the $n$ input
distributions, must produce sample outcome from $X$ that is
distributed exactly according to $f(\mu_1,\ldots,\mu_n)$.

Producing an outcome that is approximately drawn according to the
desired distribution can typically be done from estimates of the
expectations formed from sample averages (a.k.a., Monte Carlo
sampling).  On the other hand, exact implementation of many natural
functions $f$ is either impossible for information theoretic reasons
or requires sophisticated techniques.  Impossibility generally
follows, for example, when $f$ is discontinuous.  The literature on
\emph{Bernoulli Factories} (e.g., \citet{keane1994bernoulli}), which
inspires our generalization to the expectations from samples model and
provides some of the basic building blocks for our results, considers
the special case where the input distribution and output distribution
are both Bernoullis (i.e., supported on $\{0,1\}$).

We propose and solve two fundamental problems for
the expectations from samples model.  The first problem considers the
biases $\mathbf p = (p_1,\ldots,p_m)$ of $m$ Bernoulli random
variables as the marginal probabilities of a distribution on
$\{1,\ldots,m\}$ (i.e., ${\mathbf p}$ satisfies $\sum_i p_i = 1$) and
asks to sample from this distribution.  We develop an algorithm that
we call the Bernoulli Race to solve this problem.

The second problem corresponds to the ``soft maximum'' problem given
by a regularizer that is a multiple $1/\lambda$ of the Shannon entropy
function $H({\mathbf p}) = -\sum_i p_i \log p_i$.  The marginal
probabilities on outcomes that maximize the expected value of the
distribution over outcomes less the cost of the negative entropy regularizer
are given by exponential weights i.e., the function
outputs $i$ with probability proportional to $e^{\lambda p_i}$ (this is a standard
  relationship that has, for example, been employed in previous work
  in mechanism design, e.g., \citet{HuangKannan-12}).  A
straightforward exponentiation and then reduction to the Bernoulli
Race above does not have polynomial sample complexity.  We develop an
algorithm that we call the Fast Exponential Bernoulli Race to solve this
problem.

\paragraph{Black-box Reductions in Mechanism Design.}

A special case of the problem that we must solve to apply the standard
approach to black-box reductions is the \emph{single-agent
  multiple-urns problem}. In this setting, a single agent faces a set
$X$ of urns, and each urn contains a random object whose distribution
is unknown, but can be sampled. The agent's type determines his
utility for each object; fixing this type, urn $i$ is associated with
a random real-valued reward with unknown expectation $\mu_i$.  Our
goal is to allocate the agent his favorite urn, or close to it.

As described above, incentive compatibility requires an algorithm for
selecting a high-value urn that is maximal-in-range.  If we could
exactly calculate the expected values $\mu_1,\ldots,\mu_n$ from the
agent's type, this problem is trivial both algorithmically and from a
mechanism design perspective: simply solicit the agent's type $t$ then
allocate him the urn with the maximum $\mu_i=\mu_i(t)$.  As described
above, with only sample access to the expected values of each urn, we
cannot implement the exact maximum.  Our solution is to apply the Fast
Exponential Bernoulli Race as a solution to the regularized
maximization problem in the expectations from samples model.  This
algorithm -- with only sample access to the agent's values for each
urn -- will assign the agent to a random urn with a high expected
value and is maximal-in-range.

The multi-agent reduction from Bayesian mechanism design to Bayesian
algorithm design of \citet{HKM-11,HKM-15} is based on solving a
matching problem between multiple
agents and outcomes, where an agent's value for an outcome is the
expectation of a random variable which can be accessed only through
sampling. We should also assert that \citet{BH-11} independently discovered a
  similar reduction based on solving a fractional assignment
  problem. Their reduction applies to finite, discrete type spaces and
  is approximately Bayesian incentive compatible.  Specifically, this problem generalizes the above-described
single-agent multiple-urns problem to the problem of matching agents
to urns with the goal of approximately maximizing the total weight of
the matching (the social welfare).  Again, for incentive compatibility
we require this expectations from samples algorithm to be maximal-in-range from each agent's perspective.  Using methods from
\citeauthor{agrawal2015fast}'s \citeyearpar{agrawal2015fast} work on
stochastic online convex optimization, we reduce this matching problem
to the single-agent multiple-urns problem.
%% We use methods from convex
%% online stochastic optimization from \citet{agrawal2015fast} to reduce this matching problem to the single-agent multiple-urns problem.

As stated in the opening paragraph, our main result -- obtained
through the approach outlined above -- is a polynomial time reduction
from Bayesian incentive compatible mechanism design to Bayesian
algorithm design.  The analysis assumes that agents' values are
normalized to the $[0,1]$ interval and gives additive loss in the
welfare.  The reduction is an approximation scheme and the dependence
of the runtime on the additive loss is inverse polynomial.  The
reduction depends polynomially on a suitable notion of the size of the
space of agent preferences.  For example, applied to environments
where agents have preferences that lie in high-dimensional spaces, the
runtime of the reduction depends polynomially on the number of points
necessary to approximately cover each agent's space of preferences.
More generally, the bounds we obtain are polynomial in the bounds of
\citet{HKM-11,HKM-15} but the resulting mechanism, unlike in the
proceeding work, is exactly Bayesian incentive compatible.

\paragraph{Organization.} 
The organization of the paper separates the development of the
expectations from samples model and its application to black-box
reductions in Bayesian mechanism design.  Section~\ref{sec:prelim-bf}
introduces Bernoulli factories and reviews basic results from the
literature.  Section~\ref{sec:racing} defines two central problems in
the expectations from samples model, sampling from outcomes with
linear weights and sampling from outcomes with exponential weights,
and gives algorithms for solving them.  We return to mechanism design
problems in Section~\ref{sec:single} and solve the single-agent
multiple urns problem.  In Section~\ref{sec:exact-BIC-reduction} we
give our main result, the reduction from Bayesian mechanism design to
Bayesian algorithm design.

%\bibliographystyle{apalike}
%\bibliography{refs}
%\end{document}

%Incentive compatible mechanism design 

%exact algorithms from samples

%cannot assume agents are willing to approximate

%missed goal: remove the epsilon for any epsilon-BIC mechanism.

%example: urns problem

%%% Local Variables: 
%%% mode: latex
%%% TeX-master: "main"
%%% End: 

\section{Basics of Bernoulli Factories}
\label{sec:prelim}
\label{sec:prelim-bf}

We use the terms {\em Bernoulli} and {\em coin} to refer to
distributions over $\{0,1\}$ and $\{\texttt{heads},\texttt{tails}\}$,
interchangeably.  The Bernoulli factory problem is about generating new
coins from old ones.

\begin{definition}[\citealp{keane1994bernoulli}]
\label{defn:berfactory}
Given function $f : (0,1) \rightarrow (0,1)$, the \emph{Bernoulli
  factory} problem is to output a sample of a Bernoulli variable with
bias $f(p)$ (i.e.\@ an $f(p)$-coin), given black-box access to
independent samples of a Bernoulli distribution with bias $p \in
(0,1)$ (i.e.\@ a $p$-coin).
\end{definition}   

To illustrate the Bernoulli factory model, consider the examples of
$f(p)=p^2$ and $f(p)=e^{p-1}$. For the former one, it is enough to
flip the $p$-coin twice and output $1$ if both flips are $1$, and
$0$ otherwise. For the latter one, the Bernoulli factory is still
simple but more interesting: draw $K$ from the Poisson distribution
with parameter $\lambda = 1$ (remind that the Poisson distribution with parameter $\lambda$ has probability of $K = k$ as
$\lambda^ke^{-\lambda}/k!$), flip the $p$-coin $K$ times and output
$1$ if all coin flips were $1$, and $0$ otherwise (see below).

The question of characterizing functions $f$ for which there is an
algorithm from sampling $f(p)$-coins from $p$-coins has been the main
subject of interest in this
literature~(e.g., \citet{keane1994bernoulli,nacu2005fast}). In particular,
\citet{keane1994bernoulli} provides necessary and sufficient
conditions for $f$, under which an algorithm for the Bernoulli factory
exists. Moreover, \citet{nacu2005fast} suggests an algorithm for
simulating an $f(p)$-coin based on polynomial envelopes of $f$. The
canonical challenging problem of Bernoulli factories -- and a
primitive in the construction of more general Bernoulli factories --
is the {\em Bernoulli Doubling} problem: $f(p) = 2p$ for $p \in
(0,1/2)$.  See \citet{latuszynski2010bernoulli} for a survey on this
topic.

Questions in Bernoulli factories can be generalized to multiple input
coins.  Given $f:(0,1)^m\rightarrow (0,1)$, the goal is sample from a
Bernoulli with bias $f(p_1,\ldots,p_m)$ given sample access to $m$
independent Bernoulli variables with unknown biases
$\mathbf{p}=(p_1,\ldots,p_m)$.  Linear functions
$f$ were studied and solved by \citet{Huber2015OptimalLB}.  For
example, the special case $m=2$ and $f(p_1,p_2) = p_1+p_2$, a.k.a.,
{\em Bernoulli Addition}, can be solved by reduction to the Bernoulli
Doubling problem (formalized below).

Questions in Bernoulli factories can be generalized to allow input
distributions over real numbers on the unit interval $[0,1]$ (rather
than Bernoullis over $\{0,1\}$).  In this generalization the question
is to produce a Bernoulli with bias $f(\mu)$ with sample access to
draws from a distribution supported on $[0,1]$ with expectation $\mu$.
These problems can be easily solved by reduction to the Bernoulli
factory problem:
\begin{enumerate}
\item[0.] {\em Continuous to Bernoulli}: Can implement Bernoulli with
  bias $\mu$ with one sample from distribution $\mathcal D$ with
  expectation $\mu$.  Algorithm:
\begin{itemize}
\item Draw $Z \sim {\mathcal D}$ and $P \sim \bern{Z}$.
\item Output $P$.   
\end{itemize}
\end{enumerate}
Below are enumerated the important building blocks for Bernoulli factories.
%(See Appendix~\ref{app:omitted} for proofs.)
%
\begin{enumerate}
\item {\em Bernoulli Down Scaling}: Can implement $f(p) = \lambda \cdot p$ for $\lambda \in [0,1]$ with one sample from $\bern{p}$. Algorithm:
\begin{itemize}
\item Draw $\Lambda \sim \bern{\lambda}$ and $P \sim \bern{p}$.
\item Output $\Lambda \cdot P$ (i.e., $1$ if both coins are $1$, otherwise $0$).
\end{itemize}
\item {\em Bernoulli Doubling:} Can implement $f(p) = 2p$ for $p \in (0,1/2 - \delta]$ with $O(1/\delta)$ samples from $\bern{p}$ in expectation.  The algorithm is complicated, see \citet{nacu2005fast}.
\item {\em Bernoulli Probability Generating Function:} Can implement $f(p) = \ex{k\sim\mathcal{D}}{p^{k}}$ for distribution $\mathcal{D}$ over non-negative integers with $\ex{K\sim\mathcal{D}}{K}$ samples from $\bern{p}$ in expectation.  Algorithm:
\begin{itemize}
 \item Draw $K\sim \mathcal{D}$ and $P_1,\ldots,P_K \sim \bern{p}$ (i.e., $K$ samples).
 \item Output $\prod_i P_i$ (i.e., $1$ if all $K$ coins are $1$, otherwise $0$).
\end{itemize}

\item {\em Bernoulli Exponentiation:} Can implement $f(p) =
  \exp(\lambda (p-1))$ for $p \in [0,1]$ and non-negative constant
  $\lambda$ with $\lambda$ samples from $\bern{p}$ in expectation.  Algorithm:
  Apply the Bernoulli Probability Generating Function algorithm for the
  Poisson distribution with parameter $\lambda$.

\item {\em Bernoulli Averaging:} Can implement $f(p_1,p_2) = (p_1 +
  p_2)/2$ with one sample from $\bern{p_1}$ or $\bern{p_2}$.
  Algorithm:
\begin{itemize}
\item Draw $Z \sim \bern{1/2}$, $P_1 \sim \bern{p_1}$, and $P_2 \sim \bern{p_2}$.
\item Output $P_{Z+1}$.
\end{itemize}

\item {\em Bernoulli Addition:} Can implement $f(p_1,p_2) = p_1 + p_2$
  for $p_1 + p_2 \in [0,1-\delta]$ with $O(1/\delta)$ samples from
  $\bern{p_1}$ and $\bern{p_2}$ in expectation.  Algorithm: Apply
  Bernoulli Doubling to Bernoulli Averaging.
\end{enumerate}

It may seem counterintuitive that Bernoulli Doubling is much more
challenging that Bernoulli Down Scaling.  Notice, however, that for a
coin with bias $p=1/2$, Bernoulli Doubling with a finite number of
coin flips is impossible.  The doubled coin must be deterministically
heads, while any finite sequence of coin flips of $\bern{1/2}$ has
non-zero probability of occuring. On the other hand a coin with
probability $p=1/2 - \delta$ for some small $\delta$ has a similar
probability of each sequence but Bernoulli Doubling must sometimes
output tails.  Thus, Bernoulli Doubling must require a number of coin
flips that goes to infinity as $\delta$ goes to zero.

\section{The Expectations from Samples Model}
\label{sec:racing}

The expectations from samples model is a combinatorial generalization
of the Bernoulli factory problem.  The goal is to select an outcome
from a distribution that is a function of the expectations of a set of
input distributions.  These input distributions can be accessed only
by sampling.

\begin{definition}
\label{defn:expectations-from-samples}
Given function $f : (0,1)^n \rightarrow \Delta(X)$ for domain
$X$, the \emph{expectations from samples} problem is to output a
sample from $f({\boldsymbol \mu})$ given black-box access to independent
samples from $n$ distributions supported on $[0,1]$ with expectations
${\boldsymbol \mu} = (\mu_1,\ldots,\mu_n) \in (0,1)^n$.  
\end{definition}

Without loss of generality, by the Continuous to Bernoulli construction
of Section~\ref{sec:prelim-bf}, the input random variables can be
assumed to be Bernoullis and, thus, this expectations of samples model
can be viewed as a generalization of the Bernoulli factory question to
output spaces $X$ beyond $\{0,1\}$.  In this section we propose and
solve two fundamental problems for the expectations of samples model.
In these problems the outcomes are the a finite set of $m$ outcomes $X
= \{1,\ldots,m\}$ and the input distributions are $m$ Bernoulli
distributions with biases $\mathbf p = (p_1,\ldots,p_m)$.

In the first problem, biases correspond to the marginal probabilities
with which each of the outcomes should be selected.  The goal is to
produce random $i$ from $X$ so that the probability of $i$ is exactly
its marginal probability $p_i$.  More generally, if the biases do not
sum to one, this problem is equivalently the problem of {\em random
  selection with linear weights}.

The second problem we solve corresponds to a regularized maximization
problem, or specifically {\em random selection from exponential
  weights}.  For this problem the baiases of the $m$ Bernoulli input
distributions correspond to the weights of the outcomes.  The goal is
to produce a random $i$ from $X$ according to the distribution given
by exponential weights, i.e., the probability of selecting $i$ from
$X$ is $e^{\lambda p_i}/ \sum_j e^{\lambda p_j}$.

\subsection{Random Selection with Linear Weights}
\label{sec:bernoulli-race}

\begin{definition}[\textbf{Random Selection with Linear Weights}]
\label{def:ber-race}
The {\em random selection with linear weights} problem is to sample
from the probability distribution $f(\vbf)$ defined by $\prob[I \sim
  f(\vbf)]{I = i} = v_i / \sum_j v_j$ for each $i$ in $\{1,\ldots,m\}$
with only sample access to distributions with expectations $\vbf =
(v_1,\ldots,v_m)$.
\end{definition}

We solve the random selection with linear weights problem by an
algorithm that we call the {\em Bernoulli race}
(Algorithm~\ref{alg:ber-race1}).  The algorithm repeatedly picks a
coin uniformly at random and flips it. The winning coin is the first
one to come up heads in this process.

\begin{algorithm}[ht]
\small
\algblock[Name]{Start}{End}
\algblockdefx[NAME]{START}{END}%
[2][Unknown]{Start #1(#2)}%
{Ending}
\algblockdefx[NAME]{}{OTHEREND}%
[1]{Until (#1)}
 \caption{ Bernoulli Race
     \label{alg:ber-race1}}
      \begin{algorithmic}[1]
        \State{\textbf{input}}~ sample access to $m$ coins with biases $v_1,\ldots,v_m$.
        \Loop
        \State Draw $I$ uniformly from $\{1,\ldots,m\}$ and draw $P$ from input distribution $I$.  
        \State If $P$ is $\texttt{heads}$ then output $I$ and halt.
        \EndLoop
      \end{algorithmic}
\end{algorithm}

\begin{theorem}
\label{thm:ber-race}
The Bernoulli Race (Algorithm~\ref{alg:ber-race1}) samples with linear
weights (Definition~\ref{def:ber-race}) with an expected
${m}/{\sum_i v_i}$ samples from input distributions with
biases $v_1,\ldots,v_n$.
\end{theorem}
\proof
At each iteration, the algorithm terminates if the flipped coin
outputs $1$ and iterates otherwise.  Since the coin is chosen
uniformly at random, the probability of termination at each iteration
is $\frac{1}{m}\sum_{i}v_i$. The total number of iterations (and
number of samples) is therefore a geometric random variable with
expectation ${m}/{\sum_{i}v_i}$.
%  The total number of iterations is therefore a geometric random variable conditioned on not having terminated beforehand. Therefore, the number of iter one can look at the total number of iterations of the algorithm as a geometric random variable with parameter $\frac{\sum_{i\in [m]}v_i}{m}$, and so we have:
% \begin{equation}
% \ex{}{\textrm{number of sample queries}}=\ex{}{\textrm{number of iterations before termination}}=\frac{m}{\sum_{i\in [m]}v_i}.
% \end{equation}

The selected outcome also follows the desired distribution, as shown below.
\begin{align*}
\prob{\text{$i$ is selected}} &= \sum_{k=1}^\infty \prob{\text{$i$ is selected at time $k$}}\,\prob{\text{algorithm reaches time $k$}}\\
&=\frac{v_i}{m}\sum_{k=1}^\infty \left (1-\frac{1}{m}\sum\nolimits_j v_j\right)^{k-1}=\frac{\frac{v_i}{m}}{\frac{1}{m}\sum\nolimits_j v_j} = \frac{v_i}{\sum\nolimits_{j} v_j}.\qedhere
\end{align*}
\endproof

\subsection{Random Selection with Exponential Weights}
\label{sec:logistic}

\begin{definition}[\textbf{Random Selection with Exponential Weights}] 
\label{defn:log-ber-race}
For parameter $\lambda > 0$, the {\em random selection with
  exponential weights} problem is to sample from the probability
distribution $f(\vbf)$ defined by $\prob[I \sim f(\vbf)]{I = i} =
\exp(\lambda v_i) / \sum_j \exp(\lambda v_j)$ for each $i$ in
$\{1,\ldots,m\}$ with only sample access to distributions with
expectations $\vbf = (v_1,\ldots,v_m)$.
\end{definition}

The {\em Basic Exponential Bernoulli Race}, below, samples from the
exponential weights distribution.  The algorithm follows the paradigm
of picking one of the input distributions, exponentiating it, sampling
from the exponentiated distribution, and repeating until one comes up
heads.  While this algorithm does not generally run in polynomial
time, it is a building block for one that does.

\begin{algorithm}[ht]
\small
\algblock[Name]{Start}{End}
\algblockdefx[NAME]{START}{END}%
[2][Unknown]{Start #1(#2)}%
{Ending}
\algblockdefx[NAME]{}{OTHEREND}%
[1]{Until (#1)}
 \caption{The Basic Exponential Bernoulli Race (with parameter $\lambda > 0$)
     \label{alg:ber-lograce-nonpoly}}
      \begin{algorithmic}[1]
        \State{\textbf{input}}~ Sample access to $m$ coins with biases
        $v_1,\ldots,v_m$.  \State For each $i$, apply Bernoulli
        Exponentiation to coin $i$ to produce coin with bias
        $\tilde{v}_i = \exp(\lambda(v_i-1))$.  \State Run the
        Bernoulli Race on the coins with biases $\tilde{\mathbf v} =
        (\tilde{v}_1,\ldots,\tilde{v}_m)$.
      \end{algorithmic}
\end{algorithm}

\begin{theorem}
\label{thm:ber-lograce-nonpoly}
The Basic Exponential Bernoulli Race (Algorithm~\ref{alg:ber-lograce-nonpoly})
samples with exponential weights (Definition~\ref{defn:log-ber-race})
with an expected $\lambda m e^{\lambda (1-\vmax)}$ 
samples from input distributions with biases $v_1,\ldots,v_n$ and $\vmax = \max_i v_i$.
\end{theorem}

\proof
The correctness and runtime follows from the correctness and runtimes
of Bernoulli Exponentiation and the Bernoulli Race.
\endproof

\subsection{The Fast Exponential Bernoulli Race}

\label{sec:fast-exponential-bernoulli-race}

Sampling from exponential weights is typically used as a ``soft
maximum'' where the parameter $\lambda$ controls how close the
selected outcome is to the true maximum.  For such an application,
exponential dependence on $\lambda$ in the runtime would be
prohibitive.  Unfortunately, when $\vmax$ is bounded away from one,
the runtime of the Basic Logistic Bernoulli Race
(Algorithm~\ref{alg:ber-lograce-nonpoly};
Theorem~\ref{thm:ber-lograce-nonpoly}) is exponential in $\lambda$.  A
simple observation allows the resolution of this issue: the
exponential weights distribution is invariant to any uniform additive
shift of all weights.  This section applies this idea to develop the
{\em Fast Logistic Bernoulli Race}.

Observe that for any given parameter $\epsilon$, we can easily
implement a Bernoulli random variable $Z$ whose bias $z$ is within an
additive $\epsilon$ of $\vmax$. Note that, unlike the other
algorithms in this section, a precise relationship between $z$ and
$v_1,\ldots,v_m$ is not required.
\begin{lemma}\label{lem:vmax-est}
For parameter $\epsilon \in (0,1]$, there is an algorithm for sampling
  from a Bernoulli random variable with bias $z \in [\vmax - \epsilon, \vmax + \epsilon]$, where $\vmax=\max_i v_i$, with $O(\frac{m}{\epsilon^2} \cdot
  \log(\frac{m}{\epsilon}))$ samples from input distributions with
  biases $v_1,\ldots,v_m$.
\end{lemma}
\proof
The algorithm is as follows: Sample $\frac{4}{\epsilon^2} \log
(\frac{4 m}{\epsilon})$ times from each of the $m$ coins, let
$\hat{v}_i$ be the empirical estimate of coin $i$'s bias obtained by
averaging, then apply the Continuous to Bernoulli algorithm
(Section~\ref{sec:prelim-bf}) to map $\hat{v}_{\max} = \max_i
\hat{v}_i$ to a Bernoulli random variable.

Standard tail bounds (e.g., Chernoff-Hoeffding bound) imply that $|\hat{v}_{\max} - \vmax| <
\epsilon/2$ with probability at least $1-\epsilon/2$, and therefore
$z= \Ex{}{\hat{v}_{\max}} \in [\vmax - \epsilon, \vmax + \epsilon]$.
\endproof

Since we are interested in a fast logistic Bernoulli race as $\lambda$
grows large, we restrict attention to $\lambda > 4$.  We set $\epsilon
= 1/\lambda$ in the estimation of $\vmax$ (by
Lemma~\ref{lem:vmax-est}). This estimate will be used to boost the
bias of each distribution in the input so that the maximum bias is at
least $1-3\epsilon$.  The boosting of the bias is implemented with
Bernoulli Addition which, to be fast, requires the cumulative bias be
bounded away from one.  Thus, the probabilities are scaled down by a
factor of $1-2\epsilon>1/2$ (due to the fact that $\lambda>4$); this scaling is subsequently counterbalanced
by adjusting the parameter $\lambda$.  The formal details are given
below.

\begin{algorithm}[ht]
\small
\algblock[Name]{Start}{End}
\algblockdefx[NAME]{START}{END}%
[2][Unknown]{Start #1(#2)}%
{Ending}
\algblockdefx[NAME]{}{OTHEREND}%
[1]{Until (#1)}
 \caption{ Fast Exponential Bernoulli Race  (with parameter $\lambda > 4$) 
     \label{alg:ber-lograce-poly}}
      \begin{algorithmic}[1]
        \State{\textbf{input}}~ Sample access to $m$ coins with biases $v_1,\ldots,v_m$.
        \State Let $\epsilon = 1/\lambda$.

        \State Construct a coin with bias $z \in [\vmax - \epsilon,
          \vmax + \epsilon]$ (from
        Lemma~\ref{lem:vmax-est}). \label{z1}

        \State Apply Bernoulli Down Scaling 
        to a coin with bias $1-z$ to implement a coin with bias
        $(1-2\epsilon)(1-z)$. \label{z2}

        \State For all $i$, apply Bernoulli Down Scaling to implement
        a coin with bias $(1-2 \epsilon) v_i$. \label{v2}

        \State For all $i$, apply Bernoulli Addition to implement coin
        with bias $v'_i = (1-2\epsilon) v_i +(1-2 \epsilon)
        (1-z)$. \label{mutlivar}

        \State Run the Basic Exponential Bernoulli Race
        with parameter
        $\lambda' = \frac{\lambda}{1-2 \epsilon}$ on the coins with bias $v'_1,\ldots,v'_m$.
\label{last}
      \end{algorithmic}
\end{algorithm}
\begin{theorem}
\label{thm:ber-lograce-poly}
The Fast Exponential Bernoulli Race (Algorithm~\ref{alg:ber-lograce-poly})
samples with exponential weights (Definition~\ref{defn:log-ber-race})
with an expected $O(\lambda^4 m^2 \log(\lambda m))$
samples from the input distributions.
\end{theorem}
\proof
The correctness and runtime follows from the correctness and runtimes
of the Basic Exponential Bernoulli Race, Bernoulli Doubling,
Lemma~\ref{lem:vmax-est} (for estimate of $\vmax$), and the fact
$\lambda' v'_i = \lambda (v_i + 1-z)$ and the distribution given by
exponential weights is invariant to additive shifts of all weights.

A detailed analysis of the runtime follows.  Since the algorithm
builds a number of sampling subroutines in a hierarchy, we analyze the
runtime of the algorithm and the various subroutines in a bottom up
fashion.  Steps~\ref{z1} and~\ref{z2} implement a coin with bias
$(1-2\epsilon)(1-z)$ with runtime $O(\lambda^2 m \cdot \log(\lambda
m))$ per sample, as per the bound of Lemma~\ref{lem:vmax-est}. The
coin implemented in Step~\ref{v2} is sampled in constant time. Observe
that $v'_i \leq (1-2 \epsilon) (1+v_i - \vmax + \epsilon) \leq 1-
\epsilon$, and the runtime of Bernoulli Doubling implies that
$O(\lambda)$ samples from the coins of Steps \ref{z2} and \ref{v2}
suffice for sampling $\bern{v'_i}$; we conclude that a $v'_i$-coin can
be sampled in time $O(\lambda^3 m \cdot \log(\lambda m))$. Finally,
note that for $\vmax' = \max_i v'_i$, we have $\vmax' \geq 1-3
\epsilon$; Theorem~\ref{thm:ber-lograce-nonpoly} then implies that the
Basic Exponential Bernoulli Race samples at most $\lambda' m
e^{\lambda'\,3\epsilon} \leq 2 e^6 \lambda m = O(\lambda m)$ times from
the $\vbf'$-coins; we conclude the claimed runtime.  
\endproof

\section{The Single-Agent Multiple-Urns Problem}
\label{sec:single}

We investigate incentive compatible mechanism design for the
\emph{single-agent multiple-urns} problem.  Informally, the mechanism needs to assign an agent to one of many urns.  Each urn contains
objects and the agent's value for being assigned to an urn is taken in
expectation over objects from the urn.  The problem asks for an
incentive compatible mechanism with good welfare (i.e., the value of
the agent for the assigned urn).

\subsection{Problem Definition and Notations}
% types, outcomes, utility
A single agent with type $t$ from type space $\typeS{}$ desires an
object $o$ from outcome space $\outS$.  The agent's value for an
outcome $o$ is a function of her type $t$ and denoted by $v(t,o) \in
[0,1]$.  The agent is a risk-neutral quasi-linear utility maximizer
with utility $\expect[o]{v(t,o)} - p$ for randomized outcome $o$  and expected payment $p$.
%
% urns.
There are $m$ urns.  Each urn $j$ is given by a distribution $\Di{j}$
over outcomes in $\outS$.  If the agent is assigned to urn $j$ she
obtains an object from the urn's distribution $\Di{j}$.

% mechanims and objective
A mechanism can solicit the type of the agent (who may misreport if
she desires).  We further assume (1)~the mechanism has black-box access to evaluate $v(t,o)$ for any type $t$ and outcome $o$, (2)~the mechanism has sample access to
the distribution $\Di{j}$ of each urn $j$.  The mechanism may draw
objects from urns and evaluate the agent's reported value for these
objects, but then must ultimately assign the agent to a single urn and
charge the agent a payment.  The urn and payment that the agent is
assigned are random variables in the mechanism's internal
randomization and randomness from the mechanisms potential samples
from the urns' distributions.  

The distribution of the urn the mechanism assigns to an agent, as a
function of her type $t$, is denoted by $\xbf(t) =
(x_1(t),\dots,x_m(t))$ where $x_j(t)$ is the marginal probability that
the agent is assigned to urn $j$.  Denote the expected value of the
agent for urn $j$ by $v_j(t) = \expect[o\sim\Di{j}]{v(t,o)}$.  The
expected welfare of the mechanism is $\sum_j v_j(t)\, x_j(t)$.  The
expected payment of this agent is denoted by $p(t)$.  The agent's
utility for the outcome and payment of the mechanism is given by
$\sum_j v_j(t)\,x_j(t) - p(t)$.  Incentive compatibility is defined by
the agent with type $t$ preferring her outcome and payment to that
assigned to another type $t'$.

\begin{definition}
\label{d:IC}
A single-agent mechanism $(\xbf,p)$ is {\em incentive compatible} if,
for all $t,t' \in \typeS{}$:
\begin{align}
\label{eq:single-agent-ic}
 \sum\nolimits_j v_j(t)\, x_j(t) - p(t)
 &\geq  
\sum\nolimits_j v_j(t)\, x_j(t') - p(t')
\end{align}
A multi-agent mechanism is Bayesian Incentive Compatible (BIC) if
equation~\eqref{eq:single-agent-ic} holds for the outcome of the
mechanism in expectation over the truthful reports of the other agents.
\end{definition}

\subsection{Incentive Compatible Approximate Scheme}

If the agent's expected value for each urn
is known, or equivalently mechanism designer knows the distributions $\Di{j}$ for all urns $j$ rather than only sample access, this problem would be easy and  admits a trivial optimal mechanism: simply select the urn maximizing the agent's expected value $v_j(t)$  according to her reported type $t$, and charge her a payment of zero. What makes this problem
interesting is that the designer is restricted to only \emph{sample} the
agent's value for an urn. 
In this case, the following Monte-carlo adaptation of the trivial mechanism is tempting: sample from each urn sufficiently many times to obtain a close estimate $\tilde{v}_j(t)$ of $v_j(t)$ with high probability (up to any desired precision $\delta>0$), then choose the urn $j$ maximizing  $\tilde{v}_j(t)$ and charge a payment of zero. This mechanism is not incentive compatible, as illustrated by a simple example.

\begin{example}
  Consider two urns. Urn $A$ contains only outcome $o_2$, whereas $B$ two contains a mixture of outcomes $o_1$ and $o_3$, with $o_1$ slightly more likely than $o_3$. Now consider an agent who has (true) values $0$, $1$, and $2$ for outcomes $o_1$, $o_2$, and $o_3$ respectively. If this agent reports her true type, the trivial Monte-carlo mechanism --- instantiated with any desired finite degree of precision --- assigns her urn $A$ most of the time, but assigns her urn $B$ with some nonzero probability. The agent gains by misreporting her value of outcome $o_3$ as $0$, since this guarantees her preferred urn $A$.
\end{example}

The above example might seem counter-intuitive, since the trivial Monte-carlo mechanism appears to be doing its best to maximize the agent's utility, up to the limits of (unavoidable) sampling error. One intuitive rationalization is the following: an agent can slightly gain by procuring (by whatever means) more precise information about the distributions $\Di{j}$ than that available to the mechanism, and using this information to guide her strategic misreporting of her type. This raises the following question:

\paragraph{Question:} \emph{Is there an incentive-compatible mechanism for the single-agent multiple-urns problem which achieves welfare within $\epsilon$ of the optimal, and samples only  $poly(m,\frac{1}{\epsilon})$ times (in expectation) from the urns?}

We resolve the above question in the affirmative. We present approximation scheme for this problem that is
based on our solution to the problem of random selection with
exponential weights (Section~\ref{sec:logistic}).  The solution to the
single-agent multiple-urns problem is a main ingredient in the
Bayesian mechanism that we propose in
Section~\ref{sec:exact-BIC-reduction} as our black-box reduction mechanism.

To explain the approximate scheme, we start by recalling the following standard theorem in mechanism design (e.g., see \citet{G-73} and \citet{NR-01}).

\begin{theorem}
\label{t:IC}
For outcome rule $\xbf$, there exists payment rule $p$ so that
single-agent mechanism $(\xbf,p)$ is incentive compatible if and only
if $\xbf$ is maximal in range, i.e., 
$\xbf(t) \in \argmax_{\xbf'}\sum_j v_j(t)\, x_j' - c(\xbf')$, 
for some cost function $c(\cdot)$.
\end{theorem}
\begin{remark}
The ``only if'' direction of this theorem requires that the type space $\typeS{}$ be rich enough so that the induced space of values across the urns is 
$\{(v_1(t),\ldots,v_m(t)) : t \in \typeS{}\} = [0,1]^m$.
\end{remark}

The payments that satisfy Theorem~\ref{t:IC} can be easily calculated
with black-box access to outcome rule $\xbf(\cdot)$.  For a
single-agent problem, this payment can be calculated in two calls to
the function $\xbf(\cdot)$, one on the agent's reported type $t$ and
the other on a type randomly drawn from the path between the origin
and $t$.  Further discussion and details are given later in
Section~\ref{sec:implicitpayments}.  It suffices, therefore, to
identify a mechanism that samples from urns and assigns the agent to
an urn that induces an outcome rule $\xbf(\cdot)$ that is good for
welfare, i.e., $\sum_i v_j(t)\,x_j(t)$, and is maximal in range.  The
following theorem solves the problem.

\begin{theorem}
\label{thm:samu}
There is an incentive-compatible mechanism for the single-agent
multiple-urns problem which achieves an additive
$\epsilon$-approximation to the optimal welfare in expectation, and
runs in time $O(m^2 (\frac{\log m}{\epsilon})^5 )$ in expectation.
\end{theorem}

\proof
Consider the problem of selecting a distribution over urns to optimize
welfare plus (a scaling of) the Shannon entropy function, i.e.,
$\xbf(t) = \argmax_{\xbf'} v_j(t)\,x_j' - (1/\lambda) \sum_j x_j' \log
x_j'$. The additive entropy term can be interpreted as a
  negative cost vis-\`a-vis Theorem~\ref{t:IC}. It is well known
that the optimizer $\xbf(t)$ is given by exponential weights, i.e.,
the marginal probability of assigning the $j$th urn is given by $x_j(t) =
\exp(\lambda v_j(t)) / \sum_{j'} \exp(\lambda v_{j'}(t))$, a fact that can also be verified easily from the first-order conditions.  In
Section~\ref{sec:fast-exponential-bernoulli-race} we gave a polynomal
time algorithm for sampling from exponential weights, specifically,
the Fast Exponential Bernoulli Race
(Algorithm~\ref{alg:ber-lograce-poly}).  Proper choice of the
parameter $\lambda$ controls trades off faster run times with decreased
loss due to entropy term.  The entropy is maximized at the uniform
distribution $\xbf' = (1/m,\ldots,1/m)$ with entropy $\log m$.  Thus,
choosing $\lambda = \log m / \epsilon$ guarantees that the welfare is
within an additive $\epsilon$ of the optimal welfare $\max_j v_j(t)$.
The bound of the theorem then follows from the analysis of the Fast
Exponential Bernoulli Race (Theorem~\ref{thm:ber-lograce-poly}) with
this choice of $\lambda$. 
\endproof

%%% Local Variables: 
%%% mode: latex
%%% TeX-master: "main"
%%% End: 

\section{A Bayesian Incentive Compatible Black-box Reduction}
\label{sec:exact-BIC-reduction}

A central question at the interface between algorithms and economics
is on the existence of black-box reductions for mechanism design.
Given black-box access to any algorithm that maps inputs to outcomes,
can a mechanism be constructed that induces agents to truthfully
report the inputs and produces an outcome that is as good as the one
produced by the algorithm?  The mechanism must be computationally
tractable, specifically, making no more than a polynomial number of
elementary operations and black-box calls to the algorithm. 

\subsection{Basics of Bayesian mechanism design}
 Before formalizing this problem, we provide further details on Bayesian mechanism design and our set of notations in this paper, which are mostly based on those in~\citet{HKM-15}.
 
 \label{Appendix:mechnotations}
 \subsubsection{Multi-parameter Bayesian setting.}Suppose there are $\nBidder$ agents, where agent $k$ has private \emph{type} $\type{k}$ from \emph{type space} $\typeS{k}$. The \emph{type profile} of all agents is denoted by $\typeP=(\type{1},\ldots,\type{n})\in \typeS{1}\times\ldots\times \typeS{n}$. Moreover, we assume types are drawn independently from known prior distributions. For agent $k$, let $\prior{k}$ be the distribution of $\type{k}\in\typeS{k}$ and $\priorJ=\prior{1}\times\ldots\times\prior{n}$ be the joint distribution of types. Suppose there is an \emph{outcome space} denoted by $\outS$. Agent $k$ with type $\type{k}$ has valuation $v(\type{k},o)$ for outcome $o\in \outS$, where $v:(\typeS{1}\cup\ldots\cup\typeS{n})\times \outS\rightarrow [0,1]$ is a fixed function. Note that we assume agent values are non-negative and bounded, and without loss of generality in $[0,1]$.  Finally, we allow charging agents with non-negative money \emph{payments} and we assume agents are \emph{quasi-linear}, i.e., an agent with private type $t$ has \emph{utility} $u=v(t,o)-p$ for the outcome-payment pair $(o,p)$.  
 
\subsubsection{Algorithms, mechanisms and interim rules.} An \emph{allocation algorithm} $\alloc$ is a mapping from type profiles $\typeP{}$ to outcome space $\outS$. A (direct revelation) mechanism $\mech$ is a pair of \emph{allocation rule} and \emph{payment rule} $(\alloc,\payP)$, in which $\alloc$ is an allocation algorithm and $\payP=(\pay{1},\ldots,\pay{n})$ where each $\pay{k}$ (denoted by the payment rule for agent $k$) is a mapping from type profiles $\typeP$ to $\mathbb{R}_{+}\cup\{0\}$. 

One can think of the interaction between strategic agents and a mechanism as following: agents submit their \emph{reported types} $\bidP=(\bid{1},\ldots,\bid{n})$ and then the mechanism $\mech$ picks the outcome $o=\alloc(\bidP)$ and charges each agent $k$ with its payment $\pay{k}(\bidP)$. We also consider \emph{interim allocation rule}, which is the allocation from the perspective of one agent when the other agent's reported types are drawn from their prior distribution. More concretely, we abuse notation and define $\alloc^k(s_k)\triangleq \alloc(s^{k},\typeO{-k})$ to be the distribution over outcomes induced by $\alloc$ when agent $k$'s type is $s^{k}$ and other agent types are drawn from $\priorJ^{-k}$. Similarly, for agent $k$ we define \emph{interim payment rule} $\pay{k}(s^{k})\triangleq\ex{\typeO{-k}\sim\priorJ^{-k}}{\pay{k}(s^{k},\typeO{-k})}$, and \emph{interim value} $v^k(s^{k})\triangleq \ex{\typeO{-k}\sim\priorJ^{-k}}{v(s^{k},\alloc^k(s^{k},\typeO{-k}))}$. In most parts of this paper,  we focus only on one agent, e.g. agent $k$, and we just work with the interim allocation algorithm $\alloc^{k}(.)$. When it is clear from the context, we drop the agent's superscript, and therefore  $\alloc(s)$ denotes the distribution over outcomes induced by $\alloc(s,\typeO{-k})$ when $\typeO{-k}\sim \priorJ^{-k}$.

\subsubsection{Bayesian and dominant strategy truthfulness.} We are only interested in designing mechanisms that are \emph{interim truthful}, i.e., every agent bests of by reporting her true type  assuming all other agent's reported types are drawn independently from their prior type distribution. More precisely, a mechanism $\mech$ is \emph{ Bayesian Incentive Compatible (BIC)} if for all agents $k$, and all types $s^k,\type{k}\in\typeS{k},$
\begin{equation}
\ex{\typeO{-k}\sim\priorJ^{-k}}{v(\type{k},\alloc^k(\type{k}))}-\pay{k}(\type{k})\geq \ex{\typeO{-k}\sim\priorJ^{-k}}{v(\type{k},\alloc^k(s^k))}-\pay{k}(s^k)
\end{equation}
As a stronger notion of truthfulness than Bayesian truthfulness, one can consider \emph{dominant strategy truthfulness}. More precisely, a mechanism $\mech$ is \emph{Dominant Strategy Incentive Compatible (DSIC)} if for all agents $k$, all types $s^k,\type{k}\in\typeS{k}$ and all type profiles $\typeO{-k}\in \typeS{-k}$, 
\begin{equation}
{v(\type{k},\alloc(\typeP))}-\pay{k}(\typeP)\geq {v(\type{k},\alloc(s^k,\typeO{-k}))}-\pay{k}(s^k,\typeO{-k})
\end{equation} 
Moreover, an allocation algorithm $\alloc$ is said to be BIC (or DSIC) if there exists a payment rule $\payP$ such that $M=({\alloc},{\payP})$ is a BIC (or DSIC) mechanism. Throughout the paper, we use the terms Bayesian (or dominant strategy) truthful and Bayesian  (or dominant strategy) incentive compatible interchangeably. For randomized mechanisms,  DSIC and BIC solution concepts are defined by considering expectation of utilities of agents over mechanism's internal randomness.

\subsubsection{Social welfare.} 
We are considering mechanism design for maximizing \emph{social welfare}, i.e. the sum of the utilities of agents and the mechanism designer. For quasi-linear agents, this quantity is in fact sum of the valuations of the agents under the outcome picked by the mechanism. For the allocation algorithm $\alloc$, we use the notation $\val({\alloc})$ for the expected welfare of this allocation and $\valk{k}(\alloc)$ for the expected value of agent $k$ under this allocation, i.e., $\val(\alloc)\triangleq\ex{\typeP\sim\priorJ}{\sum_{k}v(\type{k},\alloc(\typeP))}$ and $\valk{k}(\alloc)\triangleq \ex{\typeP\sim\priorJ}{v(\type{k},\alloc(\typeP))}$. 

\subsection{Bayesian black-box reductions} 

A line of research initiated by \citet{HL-10,HL-15} demonstrated that,
for the welfare objective, Bayesian black-box reductions can exist. The constructed mechanism is expected to be an approximation scheme; for any $\epsilon$ the reduction gives a mechanism that is Bayesian incentive
compatible (Definition~\ref{d:IC}) and obtains a welfare that is no smaller by an additive $\epsilon$ than the algorithm's welfare in expectation. More accurately, we define the following problem.

\begin{definition}[{BIC black-box reduction problem}]
Given black-box oracle access to an allocation algorithm $\alloc$, simulate a  Bayesian incentive compatible allocation algorithm $\tilde{\alloc}$ that approximately preserves welfare, i.e. for every agent $\texttt{a}$, $\valk{\texttt{a}}(\tilde{\alloc})\geq\valk{\texttt{a}}(\alloc)-\epsilon$, and runs in time $\textrm{poly}(n,\frac{1}{\epsilon})$.
\end{definition} 

In this literature,
\citet{HL-10,HL-15} solve the case of single-dimensional agents and
\citet{HKM-11,HKM-15} solve the case of multi-dimensional agents with
discrete type spaces.  For the relaxation of the problem where 
only approximate incentive compatibility is required, \citet{BH-11}
solve the case of multi-dimensional agents with discrete type space, and \citet{HKM-11,HKM-15} solve the general case by (1) achieving exact BIC for discrete type spaces, and (2) achieving approximate BIC for general multi-dimensional type spaces. These reductions
are approximation schemes that are polynomial in the number of agents,
the desired approximation factor, and a measure of the size of the
agents' type spaces (i.e., its dimension).

Notably, one could also consider approximately preserving objectives other than welfare. However, \citet{CIL-12} have shown that BIC black-box reductions for the makespan objective cannot be computationally efficient in general. As another important note, in the
Bayesian setting, agents' types are drawn from a distribution.  The
original algorithm ideally obtains good welfare for types drawn from this distribution in expectation; although this assumption is not necessary for the reduction to work, the black-box reduction in algorithmic mechanism design makes more sense when the algorithm is assumed to obtain good welfare in such a Bayesian sense.

\subsection{Surrogate Selection and the Replica-Surrogate Matching}
\label{sec:surrogate-selcetion}
A main conclusion of the literature on Bayesian blackbox reductions for
mechanism design is that the multi-agent problem of reducing Bayesian
mechanism design to algorithm design, itself, reduces to a
single-agent problem of {\em surrogate selection}.  Consider any agent
in the original problem and the {\em induced algorithm} with the
inputs form other agents hardcoded as random draws from their
respective type distributions.  The induced algorithm maps the type of
this agent to a distribution over outcomes.  If this distribution over
outcomes is maximal-in-range then there exists payments for which the
induced algorithm is incentive compatible (Theorem~\ref{t:IC}).  If
not, the problem of surrogate selection is to map the type of the
agent to an input to the algorithm to satisfy three properties:
\begin{itemize}
\item [(a)] The composition of surrogate selection and the induced algorithm is
maximal-in-range, 
\item [(b)] The composition approximately preserves welfare,
\item [(c)] The surrogate selection preserves the type distribution.
\end{itemize}
Condition (c), a.k.a. \emph{stationarity}, implies that fixing the
non-maximaility-of-range of the algorithm for a particular agent does
not affect the outcome for any other agents.  With such an approach
each agent's incentive problem can be resolved independently from that
of other agents.

\begin{theorem}[\citet{HKM-15}] 
The composition of an algorithm with a profile of surrogate selection
rules, that maps the profile of agent types to an input to the algorithm,
is Bayesian incentive compatible and approximately preserves the
algorithms welfare (the loss in welfare is the sum of the losses in
welfare of each surrogate selection rule).
\end{theorem}

The surrogate selection rule of \citet{HKM-15} is based on setting up
a matching problem between random types from the distribution
(replicas) and the outcomes of the algorithm on random types from the
distribution (surrogates).  The true type of the agent is one of the
replicas, and the surrogate selection rule outputs the surrogate to
which this replica is matched. Given an induced allocation algorithm $\alloc(\cdot)$, assigning a replica $r_i$ to a \emph{surrogate outcome} $\alloc(s_j)$ -- which basically is a distribution over possible outcomes in $\mathcal{O}$ that the induced
algorithm produces for a surrogate $s_j$-- produces a stochastic value equal to $v(r_i,o)$, where $o\sim\alloc(s_j)$. In the aforementioned matching problem, we think of expectations of these stochastic values, i.e., the quantities $\expect[o \sim
  \alloc(s_j)]{v(r_i,o)}$ for each $(r_i,s_j)$, as weights on the edges. Now, this approach addresses the three
properties of surrogate selection rules as:
\begin{itemize}
\item[(a)] if the matching
selected is maximal-in-range given the weights, then the composition of the surrogate
selection rule with the induced algorithm is maximal-in-range,
\item[(b)] the
welfare of the matching is the welfare of the reduction and the maximum weighted
matching approximates the welfare of the original algorithm, and
\item[(c)] any maximal matching gives a stationary surrogate selection rule. 
\end{itemize}
In fact, the main reason to consider a replica-surrogate matching rather than assigning the reported type to the maximum value surrogate outcome is to obtain both welfare preservation (when market size $m$ is large enough) and stationarity property; see~\citet{HKM-15} for more details. For a detailed discussion on why maximal-in-range matching will result in a BIC mechanism after composing the corresponding surrogate selection rule with the allocation algorithm, we refer the interested reader to look at Lemma~\ref{lem:dispreserv} and Lemma~\ref{lem:singleBIC} in Appendix~\ref{app:omitted}.

In this paper, we consider a slight generalization of the surrogate selection rule in \citet{HKM-15}, which is a family of surrogate selection rules based on \emph{many-to-one matchings with budgets}. For reasons that will be clear in Section~\ref{sec:ent-reg-matching}, this degree of freedom will critically help us to go beyond ideal computational model and obtain exact BIC blackbox reductions in the expectations from samples computational model.
\begin{definition}
The \emph{replica-surrogate matching} surrogate selection rule; for a $k$-to-$1$ matching algorithm $M$, a integer market size $m$, and budget $k$; maps a type $t$ to a surrogate type as follows:
\label{defn:rsba}
\begin{enumerate}
\item Pick the real-agent index $i^*$ uniformly at random from $\{1,\ldots,km\}$.

\item Define the \emph{replica type profile} $\rbf$, an $km$-tuple of
  types by setting $r_{i^*}=t$ and sampling the remaining $km-1$ replica
  types $\rbf_{-i^*}$ i.i.d.\@ from the type distribution $\prior{}$.

\item Sample the \emph{surrogate type profile} $\sbf$, an $m$-tuple of
  i.i.d.\@ samples from the type distribution $\prior{}$.

\item 
Run matching algorithm $M$ on the complete bipartite graph between
replicas and surrogates.

\item Output the surrogate $j^*$ that is matched to $i^*$.
\end{enumerate}
\end{definition}

The value that a replica obtains for the outcome that the induced
algorithm produces for a surrogate is a random variable.  The analysis of \citet{HKM-15} is
based on the study of an ideal computational model where the value of
any replica $r_i$ for any surrogate outcome $\alloc(s_j)$, i.e., the quantity $\expect[o \sim
  \alloc(s_j)]{v(r_i,o)}$,  is known exactly.  In this
computationally-unrealistic model and with these values as weights,
the maximum weight matching algorithm can be employed in the
replica-surrogate matching surrogate selection rule above, and it results
in a Bayesian incentive compatible mechanism.~\citet{HKM-15} analyze
the welfare of the resulting mechanism in the case where the budget is
$k=1$, prove that conditions (a)-(c) are satisfied, and give
(polynomial) bounds on the size $m$ that is necessary for the expected
welfare of the mechanism to be an additive $\epsilon$ from that of the
algorithm.  

\begin{remark}
Given a BIC allocation algorithm $\tilde{\alloc}$ through a replica-surrogate matching surrogate selection, the payments that satisfy Bayesian incentive compatibility can be easily calculated with black-box access to $\tilde{\alloc}$ as implicit payments (Section~\ref{sec:implicitpayments}). 
\end{remark}

If $M$ is maximum matching, conditions (a)-(c) clearly continue to hold for our
generalization to budget $k>1$.  Moreover, the welfare of the reduction will only weakly increase for $k>1$.  

\begin{lemma}
\label{lem:matching-monotone}
In the ideal computational model (where the value of a replica for being matched to a surrogate is given exactly) the per-replica welfare of the replica-surrogate maximum matching for $k=1$ is no larger than the per-replica welfare of the replica-surrogate maximum matching for any budget $k>1$.
\end{lemma}
\proof
Consider a non-optimal matching that groups replicas into $k$ groups
of size $m$ and finds the optimal $1$-to-$1$ matching between replicas
in each group and the surrogates.  As these are random $(k=1)-$matchings, the expected welfare of each such matching is equal to the
expected welfare of the $(k=1)-$matching.  These matchings combine to
give a feasible matching between the $mk$ replicas and $m$ surrogates.
Thus, the total expected welfare of the optimal $k$-to-$1$ matching
between replicas and surrogates is at least $k$ times the expected
welfare of the $(k=1)-$matching.  Thus, the per-replica welfare, i.e.,
normalized by $mk$, is at least the per-replica welfare of the $(k=1)-$matching. 
\endproof

% In
%this section we give an approximation scheme for solving this matching
%problem that is maximal in range. 

Our main result in the remainder of this section is an approximation scheme for the ideal reduction of \citet{HKM-15}. We replace this ideal matching
algorithm with an approximation scheme for the black-box model where replica values for surrogate outcomes can only be accessed by sampling, i.e., we only have sample access to random variables $v(r_i,o)$ for $o\sim\alloc(s_j)$.  For any $\epsilon$, we identify a $k>1$ and a polynomial (in $m$ and
$1/\epsilon$) time $k$-to-$1$ matching algorithm for the black-box
model and prove that the expected welfare of this matching algorithm (per-replica) is within an additive $\epsilon$ of the expected welfare (per-replica) of the
maximum weighted matching in the ideal model with budget $k=1$ analyzed by \citealp{HKM-15}. The welfare of the ideal model is monotone non-decreasing in budget $k$ due to Lemma~\ref{lem:matching-monotone}; therefore it will be sufficient to identify a polynomial time algorithm in the black-box model that has $\epsilon$ loss relative to the ideal model for that same budget $k$. Moreover, we show our algorithm produces
a perfect (and so maximal) matching, and therefore the surrogate selection rule is
stationary; and the algorithm is maximal-in-range for the true agent's replica, and therefore the
resulting mechanism is Bayesian incentive compatible.

%In the remainder of this section we replace this ideal matching
%algorithm with an approximation scheme for the black-box model where
%replica values for surrogate outcomes can only be accessed by
%sampling.  For any $\epsilon$ our algorithm gives an $\epsilon$
%additive loss of the welfare of the ideal algorithm with only a
%polynomial increase to the runtime.  Moreover, the algorithm produces
%a perfect (and so maximal) matching, and therefore the surrogate selection rule is
%stationary; and the algorithm is maximal-in-range for the true agent's replica, and therefore the
%resulting mechanism is Bayesian incentive compatible.

\subsection{Entropy Regularized Matching}
\label{sec:ent-reg-matching}

The main idea in this section is to figure out the right maximal-in-range allocation for the replica-surrogate matching problem, so that it approximates the maximum matching allocation by an additive $\epsilon$ loss in the per-replica welfare, and also is implementable for the black-box model discussed previously.  To this end, we define an \emph{entropy regularized} bipartite matching
problem and discuss its solution. While this solution cannot be implemented as it is for reasons that we will discuss later in this section, it is the key in having a polynomial time approximate scheme for the black-box model.

Consider a complete bipartite graph with $km$ vertices on the left-hand-side and $m$ vertices on the right-hand-side. We will refer to the
left-hand-side vertices as replicas  and the right-hand-side vertices
as surrogates. Fix a replica type profile  $\rbf$ and a surrogate type profile $\sbf$.  The weights on the edge between replica $i \in
\{1,\ldots,km\}$ and surrogate $j \in \{1,\ldots,m\}$ will be denoted
by $v_{i,j}$.  In our application to the replica-surrogate matching
defined in the previous section, the weights will be set to $v_{i,j} =
\expect[o \sim \alloc(s_j)]{v(r_i,o)}$ for $(i,j)\in[km]\times[m]$.

\begin{definition} 
\label{defn:LBM}
For weights $\vbf = [v_{i,j}]_{(i,j)\in[km]\times[m]}$, the entropy regularized matching program for parameter $\delta>0$ is:
\begin{align*}
\max_{\{x_{i,j}\}_{(i,j)\in[km]\times[m]}} & \sum\nolimits_{i,j}x_{i,j}\,v_{i,j} - \delta \sum\nolimits_{i,j} x_{i,j} \log x_{i,j},\\
\textrm{s.t.~~~~} & \sum\nolimits_{i} x_{i,j} \leq k  & \forall j\in[m], \\
              & \sum\nolimits_j x_{i,j} \leq 1  & \forall i\in[km].
\end{align*}
The optimal value of this program is denoted $\OPT(\vbf)$.
\end{definition}

The dual variables for right-hand-side constraints of the matching
polytope can be interpreted as \emph{prices} for the surrogate outcomes.  Given prices, the
\emph{utility} of a replica for a surrogate outcome given prices is the
difference between the replica's value and the price.  The following
lemma, whose proof is a direct application of Karush–Kuhn–Tucker conditions, shows that for the right choice of dual variables, the maximizer
of the entropy regularized matching program is given by exponential
weights with weights equal to the utilities.

\begin{lemma}
\label{obs:dual}
For the optimal Lagrangian dual variables $\alphabf^*\in\mathbb{R}^m$
for surrogate feasibility in the entropy regularized matching program 
(Definition~\ref{defn:LBM}), namely,
\begin{align*}
\alphabf^*&=\argmin_{\alphabf\geq {\mathbf 0}} \max\nolimits_{\xbf} \big\{ \lagrange(\xbf,\alphabf) : \ \sum\nolimits_j x_{i,j} \leq 1~,~\forall i \big\}\\
\intertext{where $\lagrange(\xbf,\alphabf)\triangleq\sum_{i,j} x_{i,j}\,v_{i,j} -
\delta\sum_{i,j} x_{i,j} \log x_{i,j} +\sum_{j}\alpha_j(k-\sum_{i}
x_{i,j})$ is the Lagrangian objective function; the optimal solution $\xbf^*$ to the primal is given by exponential weights:}
x^*_{i,j}&=\frac{\exp \left(\frac{v_{i,j}-\alpha^*_j}{\delta} \right)}
 {\sum_{j'} \exp \left(\frac{v_{i,j'}-\alpha^*_{j'}}{\delta} \right)},
 \quad \forall i,j.
\end{align*}
\end{lemma}

Lemma~\ref{obs:dual} recasts the entropy regularized matching
as, for each replica, sampling from the distribution of exponential
weights.  For any replica $i$ and fixed dual variables $\alphabf$ our
Fast Exponential Bernoulli Race (Algorithm~\ref{alg:ber-lograce-poly}) gives a polynomial time algorithm for
sampling from the distribution of exponential weights in the
expectations from samples computational model.

%% Suppose replica-agent $i$ has replica type $r$. Then his value for
%% surrogate outcome distribution $O_j$ is a random value $v(r,o_j)$,
%% where $o_j\sim O_j$. This is an instance of the single-agent
%% multiple urns problem where we have an urn for each surrogate
%% outcome distribution $O_j$. Interestingly, based on the result of
%% Section~\ref{sec:single}, we know how to efficiently implement the
%% exponential weight allocations (similar to the optimal allocation
%% described in Observation~\ref{obs:dual}) by using the logistic
%% Bernoulli race algorithm
%% (Algorithm~\ref{alg:ber-lograce-poly}). The only difference is this
%% new allocation requires an additive shift of values by dual
%% variables first, which can easily be taken care of truthfully. More
%% concretely, we have the following lemma.
\begin{lemma} 
\label{lem:logsticprice}
For replica $i$ and any prices (dual variables) $\alphabf\in[0,h]^m$, allocating a surrogate $j$
drawn from the exponential weights distribution 
\begin{equation}
\label{eq:logisticprice}
x_{i,j}=\frac{ \exp \left(\frac{v_{i,j}-\alpha_j}{\delta} \right)}
{\sum_{j'} \exp \left(\frac{v_{i,j'}-\alpha_{j'}}{\delta} \right)},
\quad \forall j\in[m],
\end{equation}
 is maximal-in-range for replica $i$, as defined in Definition~\ref{t:IC}, and this random surrogate $j$ can be sampled with $O \left( \frac{h^4
  m^{2}\log(hm/\delta) }{{\delta^4}} \right)$ samples from
replica-surrogate-outcome value distributions.
\end{lemma}

\proof
To see that the distribution is maximal-in-range when assigning surrogate outcome $j$ to replica $i$, consider the regularized welfare maximization 
$$\argmax_{\xbf'}\sum_j v_{i,j}\, x_j' - \delta\sum_{j}x'_{j}\log{x'_{j}}-\sum_{j}\alpha_jx'_{j}$$ 
for replica $i$.  By looking at the first-order conditions, similar to Lemma~\ref{obs:dual}, it is easy to see that the exponential weight distribution in (\ref{eq:logisticprice}) is the unique maximizer of this concave program.

To apply the Fast Exponential Bernoulli Race to the utilities, of the
form $v_{i,j} - \alpha_j \in [-h,1]$, we must first normalize them to
be on the interval $[0,1]$.  This normalization is accomplished by
adding $h$ to the utilities (which has no effect on the exponential
weights distribution, and therefore preserves being maximal-in-range), and then scaling by $1/(h+1)$.  The scaling
needs to be corrected by setting $\lambda$ in the Fast Exponential
Bernoulli Race (Algorithm~\ref{alg:ber-lograce-poly}) to
$(h+1)/\delta$.  The expected number of samples from the value
distributions that are required by the algorithm, per
Theorem~\ref{thm:ber-lograce-poly}, is $O(h^4m^2\log
(hm/\delta)\delta^{-4} )$. 

\endproof

If we knew the optimal Lagrangian variables $\alphabf^*$ from
Lemma~\ref{obs:dual}, it would be sufficient to define the
surrogate selection rule by simply sampling from the true agent
$i^*$'s exponential weights distribution (which is polynomial time per
Lemma~\ref{lem:logsticprice}).  Notice that the wrong values
of $\alphabf$ correspond to violating primal constraints for the
surrogates. Thus the outcome from sampling from exponential
weights for such $\alpha$ would not correspond to a matching, while remains to be maximal-in-range for each replica.
In the next section we propose a polynomial time approximation scheme
that outputs a matching that is maximal-in-range for each replica, and therefore for the true agent $i^*$, and at the same time approximates sampling
from the correct $\alphabf^*$.

%% Now, one might think that implementing an optimal solution to the RSLM
%% truthfully can be done by just using Lemma~\ref{lem:logsticprice} with
%% $\alphabf=\alphabf^*$. However, such an allocation rule cannot be
%% truthful as $\alpha^*$ depends on the entire replica type profile, and
%% therefore from the perspective of any replica-agent $i$ the allocation
%% in (\ref{eq:logisticprice}) is not a DSIC allocation.

\subsection{Online Entropy Regularized Matching}

In this section, we reduce the entropy regularized matching problem to the problem of sampling from exponential weights (as described in Lemma~\ref{lem:logsticprice}) in a sequential fashion over all replicas. Although the actual problem is indeed an offline problem, we treat it as an \emph{online problem} where replicas arrive online, but  the ordering under which they arrive is in our control. This treatment helps us to preserve the maximal-in-range property of our assignment for each replica, while guaranteeing primal feasibility and near-optimal objective value in the entropy regularized matching problem (and hence near-optimal social welfare for small enough $\delta$).

Similar to Section~\ref{sec:ent-reg-matching}, fix arbitrary profiles of replicas $\rbf$ and surrogates $\sbf$.  Now consider going over replicas $\rbf$ in a \emph{random order}, over times $i=1,\ldots,km$, and assigning them to the surrogates by sampling from the exponential weights distribution as given by Lemma~\ref{lem:logsticprice} with prices $\alphabf^{(i)}, i=1,\ldots,km$  (we will detail later how to set these prices). The basic observation is that (near-optimal) approximate dual variables $\alphabf^{(i)}, i=1,\ldots,km$ are sufficient for an online
assignment of each replica to a surrogate via
Lemma~\ref{lem:logsticprice} to obtain (near-optimal) approximations to the optimum offline
regularized matching. Moreover, such a sequential assignment will result in a maximal-in-range allocation for each replica.

How to construct a sequence of dual prices $\alphabf^{(i)}$ for $i=1,\ldots,km$ -- ideally in an online fashion-- that can play the role of near-optimal approximations to the optimal dual prices $\alphabf^*$ of the offline problem? How to preserve the feasibility of our assignment by respecting the matching constraints of  the surrogate side? To address these questions, we propose a primal-dual algorithm by borrowing techniques used in designing online algorithms for
stochastic online convex
programming problems~(\citet{agrawal2015fast,chen2015dynamic}), and stochastic online packing problems (\citet{agrawal2009dynamic,devanur2011near,badanidiyuru2013bandits,kesselheim2014primal}).

Our primal-dual online algorithm (Algorithm~\ref{alg:onlineLMB}, below)  considers
the replicas in a random order, updates the dual variables based on the current number of allocated replicas to each surrogate (\emph{dual update} step),  and allocates an available surrogate to each
arriving replica by sampling from the exponential weights distribution as given
by Lemma~\ref{lem:logsticprice} with updated dual variables (\emph{primal assignment} step). Under the hood, the dual update is essentially running a no-regret learning algorithm -- such as {exponential gradient ascent}~\citep{shalev2012online} (also known as multiplicative weights), or {follow-the-perturbed-leader}~\citep{kalai2005efficient}, or  {online mirror descent}~\citep{bubeck2014convex} --  for a specific adversarial online linear optimization problem (which we explain later). Roughly speaking, this online learning algorithm tries to learn a dual assignment that fits the primal allocation the best in terms of dual complementary slackness, or equivalently tries to minimize the objective value of the Fenchel dual program of the primal entropy regularized matching problem (cf. \citet{boyd2004convex}). For the ease of exposition, we use exponential gradient ascent for our dual updates in Algorithm~\ref{alg:onlineLMB}, but in principle it can be replaced by any online learning algorithm with the same regret guarantees.

Algorithm~\ref{alg:onlineLMB} is parameterized by
$\delta$, the scale of the regularizer; by $\eta$, the rate at which
the algorithm learns the dual variables $\alphabf$; and by scale
parameter $\gamma$. For technical reasons, our algorithm uses scaled dual prices $\gamma\alphabf^{(i)}$; we detail later why this modification is needed and how to set scale parameter $\gamma$.
%

%EXPAND THIS
%PARAGRAPH OR ADD ANOTHER TO DISCUSS RELATIONSHIP BETWEEN OUR ALGORITHM
%and ``online mirror descent'' AND ``exponential gradient''.

\floatname{algorithm}{Algorithm}
\begin{algorithm}[tbh]
\small
\algblock[Name]{Start}{End}
\algblockdefx[NAME]{START}{END}%
[2][Unknown]{Start #1(#2)}%
{Ending}
\algblockdefx[NAME]{}{OTHEREND}%
[1]{Until (#1)}
 \caption{ Online Entropy Regularized Matching Algorithm (with parameters $\delta, \eta, \gamma \in \mathbb{R}_+$)
     \label{alg:onlineLMB}}
      \begin{algorithmic}[1]
        \State {\bf input:} sample access to replica-surrogate matching instance with expected values $\{v_{i,j}\}$ for replicas $i \in \{1,\ldots,mk\}$ and surrogates $j \in \{1,\ldots,m\}$.
        \vspace{2mm}
       \State Shuffle the replicas by a uniform random permutation $\pi:[mk]\rightarrow [mk]$, hence indexed by $\pi(1),\ldots,\pi(mk)$.
%      	\State Calculate $\gamma = \OPT(\tilde{\vbf})$ for
%        $\tilde{\vbf}$ calculated as the empirical sample averages of
%        a new $km$ surrogage to $m$ replica matching.

\ForAll{$i \in \{1,\ldots,km\}$}

        \State{Let $k_j$ be the number of replicas previously matched to
        each surrogate $j$ and $J = \{j : k_j < k\}$ the set of surrogates with
        availability remaining.}

      	\State\underline{\emph{Dual update:}} set $\alphabf^{(i)}$ according to the exponential weights distribution with weights $\eta \cdot k_j$ for available surrogates
       $j \in J:~\alpha_j^{(i)}=\exp(\eta\cdot k_j)/\sum_{j'\in J}\exp(\eta\cdot k_{j'})$.

      	\State \underline{\emph{Primal assignment:}} By running the fast exponential Bernoulli race (Algorithm~\ref{alg:ber-lograce-poly}), match the arriving replica at time $i$ (i.e., replica $\pi(i)$) to an available surrogate $j \in J$ drawn according to the exponential weights distribution with weights $(v_{\pi(i),j} -
        \gamma\, \alpha_j^{(i)})/\delta$.

%% \left\{ \begin{array}{ll}
%%         1, & \mbox{if $i=\hat{\imath}_j$};\\
%%         0 ,& \mbox{otherwise}.\end{array} \right.$

\EndFor
      	\end{algorithmic}
\end{algorithm}

The final algorithm needs to satisfy four properties to be useful as a surrogate selection rule in a
polynomial time Bayesian incentive compatible blackbox reduction.  First, it needs to produce a maximal
matching so that the replica-surrogate matching surrogate selection
rule is stationary, specifically via condition (c) in Section~\ref{sec:surrogate-selcetion}.  It needs to be
maximal-in-range for the real agent (replica $i^*$).  In fact, all
replicas are treated symmetrically and allocated by sampling from an
exponential weights distribution that is maximal-in-range via
Lemma~\ref{lem:logsticprice}.  Third, it needs to have good welfare
compared to the ideal matching (no smaller than $\epsilon$ from optimal welfare in the ideal model).  Fourth, its runtime needs to be
polynomial.  The first two properties are immediate and imply the
theorem below.  The last two properties are analyzed in the next section.

\begin{theorem} 
The mechanism that maps types to surrogates via the replica-surrogate
matching surrogate selection rule with the online entropy regularized
matching algorithm (Algorithm~\ref{alg:onlineLMB}) is Bayesian
incentive compatible (truthful payments are computed implicitly from Theorem~\ref{t:IC}).
\end{theorem}

\subsection{Social Welfare Loss}

We analyze the welfare loss of the online entropy regularized matching
algorithm (Algorithm~\ref{alg:onlineLMB}) with regularizer parameter
$\delta$, learning rate $\eta$, and scale parameter $\gamma$. During the analysis, we show how to set these parameters to guarantee the per-replica expected welfare loss is at most $\epsilon$ .

\begin{theorem}
\label{thm:main-bic-reduction}
There are parameter settings for online entropy regularized matching
algorithm (Algorithm~\ref{alg:onlineLMB}) for which (1)~its per-replica
expected welfare is within an additive $\epsilon$ of the optimal
welfare of the replica surrogate matching, and (2)~given oracle access to $\alloc$, the running time of this algorithm is polynomial in $m$ and $1/\epsilon$.
\end{theorem}

\proof[\textbf{Proof overview}] To prove Theorem~\ref{thm:main-bic-reduction},  we consider the following three steps: 
\begin{itemize}
\vspace{1mm}
\item~\emph{Step \rom{1} :} We first analyze the performance  of Algorithm~\ref{alg:onlineLMB} with learning rate $\eta$ in the entropy regularized matching problem, and argue that our online algorithm and the offline optimal entropy regularized matching algorithm have nearly the same (per-replica) objective value in the convex program, up to an additive loss of $O(\eta)$. We show this result holds if the scale parameter $\gamma$ is set appropriately and $k$ is large enough (still polynomial in $m$ and $1/\eta$). 
\vspace{1mm}
\item~\emph{Step \rom{2}:}  It turns out that to obtain the result in Step~\rom{1}, we need to set $\gamma$ to be a constant approximation to the $k$-fraction of the optimal objective value of the offline convex program in Definition~\ref{defn:LBM}, and also an overestimation. We then argue how to set $\gamma$ to be such an approximation/estimation for the optimal objective value of our offline convex program
with high probability, and with efficient sampling. We do this step while preserving incentive compatibility.
\vspace{1mm}
\item~\emph{Step \rom{3}:} We argue that for small enough regularizer parameter $\delta>0$, the value of the convex objective of the offline optimal entropy regularized matching is nearly as large as the welfare of the offline optimal matching.
\vspace{1mm}
\item~\emph{Step \rom{4}:} Finally, the proof of the
theorem is finished by combining the above steps with the right
choice of parameters $\delta$ and $\eta$ (as functions of $\epsilon$ and $m$), and observing this choice guarantees (i) an additive per-replica social welfare loss of $O(\epsilon)$ with respect to any replica-surrogate $k$-to-$1$ matching, and (ii) polynomial in $m$ and $1/\epsilon$ blackbox oracle complexity (and also running time) for the final algorithm.
\end{itemize}
 \sbox{0}{\popQED}
\endproof

\proof[\textbf{Proof details of Theorem~\ref{thm:main-bic-reduction}}]
We provide the details of the above four steps below.
\sbox{0}{\popQED}
\endproof

\subsubsection*{\underline{Step \rom{1}:} Additive per-replica loss of the online entropy regularized matching algorithm.}
As before, fix arbitrary profiles of replicas $\rbf$ and surrogates $\sbf$, and hence the replica-surrogate expected values $\vbf$. Recall that the expected values $\vbf$ play the role of edge weights in our bipartite graph. Also, recall that $\OPT(\vbf)$ denotes the optimal objective value of the
entropy regularized matching program. We now prove the following proposition.
\begin{proposition} 
\label{prop:compt-analysis}
For a fixed regularizer parameter $\delta>0$, learning rate $\eta>0$,
scale parameter $\gamma>0$, budget $k\in\mathbb{N}$, and market size $m\in\mathbb{N}$ that
satisfy 
%$$\frac{m \log m }{\eta^2}\leq k \leq m^{O(1)}~~~~\textrm{and}~~~~\OPT(\vbf)/k\leq \gamma\leq O(1) \OPT(\vbf)/k~,$$

$$k\geq \frac{m \log ({m}/{\eta}) }{\eta^2}~~~~\textrm{and}~~~~\OPT(\vbf)/k\leq \gamma\leq O(1) \OPT(\vbf)/k~,$$
the online entropy regularized matching algorithm
(Algorithm~\ref{alg:onlineLMB}) obtains an objective value within an additive $O(\eta)$ of the objective value of the optimal entropy regularized matching $\OPT(\vbf)$
(Definition~\ref{defn:LBM}).
\end{proposition}
We start by showing the following technical lemma, which is going to be useful in several steps of the proof of Proposition~\ref{prop:compt-analysis}.
\newcommand{\zbf}{\mathbf{z}}
\newcommand{\ybf}{\mathbf{y}}
\begin{lemma} 
\label{lemma:technical}
Given vectors $\zbf_1,\ldots,\zbf_T\in[0,1]^d$, where $d,T\in \mathbb{N}$, and uniform random permutation  $\pi:[T]\rightarrow[T]$ over $\{1,2,\ldots,T\}$,  define $\ybf_t\triangleq \sum_{\tau=t+1}^{T}\zbf_{\pi(\tau)}/(T-t)$. Then:
$$
\sum_{t=1}^T\mathbf{E}\left[\underset{i\in[d]}{\max}\left\lvert y_{t,i}-y_{0,i}\right\rvert \right]\leq  O\left(\sqrt{T(\log T+\log d)}\right)
$$
\end{lemma}
\proof
Fix $i\in[d]$ and consider the stochastic sequence $y_{t,i}$ for $ t=1,2,\ldots,T$.  We have:
$$
\mathbf{E}\left[y_{t,i}|\pi(1),\ldots,\pi(t-1)\right]=\frac{\mathbf{E}\left[\sum_{\tau=t+1}^{T}z_{\pi(\tau),i}|\pi(1),\ldots,\pi(t-1)\right]}{T-t}=\frac{\frac{T-t}{T-t+1}\sum_{\tau=t}^{T}z_{\pi(\tau),i}}{T-t}=y_{t-1,i}.
$$
Therefore, $\{y_{t,i}\}$ is a martingale sequence with respect to random variables $\pi(1),\pi(2),\ldots,\pi(T)$. Moreover, this martingale sequence has bounded difference, simply because
$$
\lvert y_{t,i}-y_{t-1,i}\rvert=\lvert\frac{\sum_{\tau=t+1}^T z_{\tau,i}}{T-t}-\frac{\sum_{\tau=t}^T z_{\tau,i}}{T-t+1}\rvert=\frac{\lvert \sum_{\tau=t+1}^T z_{\tau,i}- (T-t)z_{t,i}\rvert}{(T-t)(T-t+1)}\leq \frac{\sum_{\tau=t+1}^T\lvert z_{\tau,i}-z_{t,i}\rvert}{(T-t+1)(T-t)}\leq \frac{1}{T-t+1}.
$$
Let $c_t\triangleq 1/(T-t+1)$. Then by using Azuma-Hoeffding concentration bound for bounded difference martingales, for every $\delta>0$ we have:
$$\textrm{Pr}\left\{\lvert y_{t,i}-y_{0,i}\rvert>\delta\right\}\leq 2 \exp\left(-\frac{\delta^2}{\sum_{\tau=1}^{t}c_{t}^2}\right)=2 \exp\left(-\frac{\delta^2}{\sum_{\tau=1}^{t}\frac{1}{(T-\tau+1)^2}}\right)\leq 2\exp\left(-\frac{\delta^2(T-t)}{2}\right)~,$$
where the last inequality holds because
$$
\sum_{\tau=1}^t \frac{1}{(T-\tau+1)^2}=\sum_{\tau=T-t+1}^T\frac{1}{\tau^2}\leq \int_{T-t}^T\frac{1}{x^2}dx=\frac{1}{T-t}-\frac{1}{T}<\frac{1}{T-t}~.$$
By applying union bound, we have:
$$\textrm{Pr}\left\{\underset{i\in[d]}{\max}\left\{\lvert y_{t,i}-y_{0,i}\rvert\right\}>\delta\right\}\leq 2d\exp\left(-\frac{\delta^2(T-t)}{2}\right)~.$$
Therefore, by setting $\delta=\sqrt{\frac{2\log(Td)}{T-t}}$, we have
$$
\mathbf{E}\left[\underset{i\in[d]}{\max}\left\lvert y_{t,i}-y_{0,i}\right\rvert \right]\leq \delta+2d\exp\left(-\frac{\delta^2(T-t)}{2}\right)=\sqrt{\frac{2\log(Td)}{T-t}}+\frac{2}{T}=O\left(\sqrt{\frac{\log(Td)}{T-t}}\right)~.
$$
Now, by summing the above upper-bound over $t=1,2,\ldots,T$, we have
$$
\sum_{t=1}^{T}\mathbf{E}\left[\underset{i\in[d]}{\max}\left\lvert y_{t,i}-y_{0,i}\right\rvert \right]\leq 1+\sum_{t=1}^{T-1}\mathbf{E}\left[\underset{i\in[d]}{\max}\left\lvert y_{t,i}-y_{0,i}\right\rvert \right]\leq O\left(\sqrt{\log(Td)}\right)\sum_{t=1}^{T-1}\frac{1}{\sqrt{T-t}}=O\left(\sqrt{T\log(Td)}\right)~,
$$
where the equation holds because $\sum_{t=1}^{T-1}\tfrac{1}{\sqrt{T-t}}\leq \int_{0}^{T-1}\frac{1}{x^{1/2}}dx=O(\sqrt{T})$.
\endproof

\proof[Proof of Proposition~\ref{prop:compt-analysis}]
Let permutation $\pi:[km]\rightarrow[km]$ denote the replica arrival ordering, meaning that replica $r_{\pi(i)}$ arrives at time $i\in[km]$. $\pi$ is a uniformly random permutation. Let ${\xbf}\in\{0,1\}^{mk\times m}$ denote the actual allocation of Algorithm~\ref{alg:onlineLMB}, that is,  ${x}_{i,j}=1$ if replica $r_{\pi(i)}$ is matched to surrogate $s_j$ and ${x}_{i,j}=0$ otherwise. We further use $\overline{\xbf}$ to denote the ``conditional matching probabilities'' of replicas to surrogates in Algorithm~\ref{alg:onlineLMB}, that is,
\begin{align*}
&\forall j\in J:~~\overline{x}_{i,j}=\frac{\exp\left(\frac{v_{\pi(i),j}-\gamma\alpha^{(i)}_j}{\delta}\right)}{\sum_{j'\in J}\exp\left(\frac{v_{\pi(i),j'}-\gamma\alpha^{(i)}_{j'}}{\delta}\right)}~~,\\
&\forall j\notin J:~~\overline{x}_{i,j}=0~.
\end{align*}
Define stopping time $\tau$ to be the first time that one of the surrogates becomes  unavailable (because all $k$ copies are matched to previous replicas), i.e., 
$$
\tau\triangleq \min \left\{\{t\in[mk]:\exists~j~~\textrm{s.t.}~~\sum_{i=1}^t{x}_{i,j}>k\}\cup\{mk+1\}\right\}~.
$$
Notice that either $\tau-1=mk$ or there exists surrogate $j$ such that $\sum_{i=1}^{\tau-1}{x}_{i,j}=k$. Moreover, define the following quantities for each $i\in[km]$:
\begin{align*}
&\alg{i}(\mathbf{v})\triangleq \sum_{j\in[m]}v_{\pi(i),j}{x}_{i,j}~~,\\
&\overline{\alg{}}_i
(\mathbf{v}) \triangleq \sum_{j\in[m]} v_{\pi(i),j}\,\overline{x}_{i,j} - \delta \sum_{j\in[m]} \overline{x}_{i,j} \log
\overline{x}_{i,j}~.
\end{align*}
Note that ${\alg{i}}(\mathbf{v}) $ is the contribution of the algorithm at time $i$ to the social welfare, and $\overline{\alg{}}_i(\mathbf{v}) $ is the (fractional) contribution of (allocation probabilities of) the algorithm at time $i$ to the convex objective of the entropy regularized matching problem. Likewise, let $\xbf^*$ denote the fractional optimum solution of the offline convex optimization for entropy regularized matching, indexed in a way that $\xbf_i^*$ assigns replica $r_{\pi(i)}$ (fractionally) to surrogates. Let $\opt_i(\mathbf{v})$ denote the contribution of $\xbf_i^*$ to the convex objective of entropy regularized matching for each $i\in[km]$, that is,
$$
\opt_i(\mathbf{v})\triangleq \sum_{j\in[m]}v_{\pi(i),j}x^*_{i,j}-\delta\sum_{j\in[m]}x^*_{i,j}\log x^*_{i,j}~.
$$
Note that the optimum objective value $\OPT(\mathbf{v})$ of the entropy regularized matching problem (see Definition~\ref{defn:LBM}) is oblivious to the ordering induced by $\pi$ and is equal to $\sum_{i=1}^{km}\opt_i(\mathbf{v}) $. For simplicity of notation, we drop the input argument $\mathbf{v}$ from $\OPT(\mathbf{v})$, $\opt_i(\mathbf{v})$, $\alg{i}(\mathbf{v})$ and $\overline{\alg{}}_i(\mathbf{v})$ in the rest of the proof.

Now consider times $i=1,2,\ldots,\tau-1$. At each time $i$, a new replica $r_{\pi(i)}$ arrives. For any given allocation $\xbf'_i =(x'_{i,1},\ldots,x'_{i,m})$ of replica $r_{\pi(i)}$ to surrogates and any given scaled prices/dual variables $\gamma\alphabf^{(i)}$ at time $i$, define the contribution of replica $r_{\pi(i)}$ to the Lagrangian objective of
Lemma~\ref{obs:dual} as
\begin{align}
\label{eq:single-alloc}
\lagrange^{(i)}(\xbf'_i,\gamma\alphabf^{(i)})&\triangleq \sum_{j\in[m]}v_{\pi(i),j}\,x'_{i,j} - \delta \sum_{j\in[m]}x'_{i,j} \log x'_{i,j} + \sum_{j\in[m]}\gamma\alpha^{(i)}_j(\tfrac{1}{m}-x'_{i,j}).
\end{align}
At each time $i=1,2,\ldots,\tau-1$, the difference between $\overline{\xbf}_i=(\overline{x}_{i,1},\ldots,\overline{x}_{i,m})$ picked by Algorithm~\ref{alg:onlineLMB} and $\xbf^*_i=(x^*_{i,1},\ldots,x^*_{i,m})$ picked by the offline optimum is that the algorithm selects its matching (conditional) probabilities with respect to dual variables $\gamma\alphabf^{(i)}$, while the offline optimum selects its fractional matching
with respect to the optimal dual variables $\alphabf^*$
(Lemma~\ref{obs:dual}).  In fact, we have: 
$$\overline{\xbf}_i \in
\argmax_{\xbf'_i \in \Delta^m} \lagrange^{(i)}(\xbf'_i,\gamma
\alphabf^{(i)})~~~~~~,~~~~~~\xbf^*_i \in
\argmax_{\xbf'_i \in \Delta^m} \lagrange^{(i)}(\xbf'_i,
\alphabf^*)$$
The optimality of $\overline{\xbf}_i$ for dual variables $\gamma\alphabf^{(i)}$--combined with equation~\eqref{eq:single-alloc}-- implies 
\begin{equation*}
\overline{\alg{}}_{i}+\sum_{j\in[m]}\gamma\alpha_j^{(i)} \big(\tfrac{1}{m}-\overline{x}_{i,j} \big)\geq \opt_i +\sum_{j\in[m]}\gamma\alpha_j^{(i)} \big(\tfrac{1}{m}-x^*_{i,j} \big)~~.
\end{equation*}
By rearranging the terms, we have
\begin{align*}
\overline{\alg{}}_i&\geq \opt_i+\gamma\alphabf^{(i)}\cdot\overline{\xbf}_i-\gamma\alphabf^{(i)}\cdot\xbf^*_i\\
&=\ex{}{\opt_{i}}+\gamma\alphabf^{(i)}\cdot\overline{\xbf}_i-\gamma\alphabf^{(i)}\cdot\ex{}{\xbf^*_i}+\left(\opt_{i}-\ex{}{\opt_i}\right)-\gamma\alphabf^{(i)}\cdot\left(\xbf^*_i-\ex{}{\xbf^*_i}\right) \\
&\overset{(1)}{\geq} \ex{}{\opt_{i}}+\gamma\alphabf^{(i)}\cdot(\overline{\xbf}_i-\frac{1}{m}\mathbf{1})+ \left(\opt_{i}-\ex{}{\opt_i}\right)-\gamma\alphabf^{(i)}\cdot\left(\xbf^*_i-\ex{}{\xbf^*_i}\right)\\
&\overset{(2)}{=} \frac{{\opt}}{km}+\gamma\alphabf^{(i)}\cdot(\overline{\xbf}_i-\frac{1}{m}\mathbf{1})+ \left(\opt_{i}-\ex{}{\opt_i}\right)-\gamma\alphabf^{(i)}\cdot\left(\xbf^*_i-\ex{}{\xbf^*_i}\right)~,
\end{align*}
where inequality~(1) holds because:
 $$\forall j\in [m]:~~\ex{}{x^*_{ij}}=\frac{\sum_{i'\in[km]}x^*_{\pi(i'),j}}{km}\leq \frac{k}{km}=\frac{1}{m}~,$$
 and equality~(2) holds because:
 $$\ex{}{\opt_i}=\frac{\sum_{i'\in[km]}\opt_{\pi(i')}}{km}=\frac{\opt}{km}~.$$
Now, suppose the observed history up to time $i$ is denoted by $\hist{i-1}$. For each $i=1,2,\ldots,\tau-1$, consider a history path $\hist{i-1}$ that leads to $i\leq \tau-1$.  By taking expectation conditioned on any such history path $\hist{i-1}$, we have: 
  \begin{align*}
 \expect{\overline{\alg{}}_i \given \hist{i-1}}
  \geq&~\frac{\opt}{km}+\gamma\expect{\alphabf^{(i)}\cdot(\overline{\xbf}_i-\frac{1}{m}\mathbf{1})\given \hist{i-1}}
   + \left(\expect{\opt_{i}\given\hist{i-1}}-\ex{}{\opt_i}\right)\\
   &-\gamma\expect{\alphabf^{(i)}\cdot\left(\xbf^*_i-\ex{}{\xbf^*_i}\right)\given \hist{i-1}}\\
   \overset{(1)}{=}&~\frac{{\opt}}{km}+\gamma\alphabf^{(i)}\cdot(\expect{\overline{\xbf}_i\given\hist{i-1}}-\frac{1}{m}\mathbf{1})
   + \left(\expect{\opt_{i}\given\hist{i-1}}-\ex{}{\opt_i}\right)\\
   &-\gamma\alphabf^{(i)}\cdot\left(\expect{\xbf^*_i\given \hist{i-1}}-\ex{}{\xbf^*_i}\right)\\
    \overset{(2)}{=}&~\frac{{\opt}}{km}+\gamma\alphabf^{(i)}\cdot(\expect{{\xbf}_i\given\hist{i-1}}-\frac{1}{m}\mathbf{1})
   + \left(\expect{\opt_{i}\given\hist{i-1}}-\ex{}{\opt_i}\right)\\
   &-\gamma\alphabf^{(i)}\cdot\left(\expect{\xbf^*_i\given \hist{i-1}}-\ex{}{\xbf^*_i}\right)\\
    =&~\frac{{\opt}}{km}+\gamma\alphabf^{(i)}\cdot({\xbf}_i-\frac{1}{m}\mathbf{1})
   + \left(\expect{\opt_{i}\given\hist{i-1}}-\ex{}{\opt_i}\right)\\
   &-\gamma\alphabf^{(i)}\cdot\left(\expect{\xbf^*_i\given \hist{i-1}}-\ex{}{\xbf^*_i}\right)+\gamma\alphabf^{(i)}\cdot\left(\expect{{\xbf}_i\given\hist{i-1}}-{\xbf}_i\right)
 \end{align*}
where equality~(1) holds as $\alphabf^{(i)}$ is only a function of history up to time $i$, and equality~(2) holds as $\expect{\overline{\xbf}_i\given\hist{i-1},\pi(i)}=\expect{{\xbf}_i\given\hist{i-1},\pi(i)}$ for any $\pi(i)$, and therefore $\expect{\overline{\xbf}_i\given\hist{i-1}}=\expect{{\xbf}_i\given\hist{i-1}}$. To simplify the calculations, we introduce the following extra notations: 
\begin{align*}
O_i&\triangleq \expect{\opt_{i}\given\hist{i-1}}-\ex{}{\opt_i}\\
Z_i&\triangleq \gamma\alphabf^{(i)}\cdot\left(\expect{\xbf^*_i\given \hist{i-1}}-\ex{}{\xbf^*_i}\right)\\
L_i&\triangleq \gamma\alphabf^{(i)}\cdot\left(\expect{{\xbf}_i\given\hist{i-1}}-{\xbf}_i\right)
\end{align*}
By summing over times $i=1,2,\ldots,\tau$, we have: 
\begin{align}
\label{eq:regret}
\sum_{i=1}^{\tau-1}\expect{\overline{\alg{}}_i \given \hist{i-1}}\geq \frac{\tau-1}{km}\opt+\gamma\sum_{i=1}^{\tau-1}g_i(\alphabf^{(i)})-\sum_{i=1}^{\tau-1}\lvert O_i\rvert -\sum_{i=1}^{\tau-1}\lvert Z_i \rvert+\sum_{i=1}^{\tau-1}L_i~,
\end{align}
where $g_i:[0,1]^m\rightarrow \mathbb{R}$ is defined as follows for each $i=1,2,\ldots,\tau-1$:
\begin{equation*}
g_i(\alphabf)\triangleq\alphabf\cdot({\xbf}_i-\frac{1}{m}\mathbf{1})
\end{equation*}

In order to bound the term $\sum_{i=1}^{\tau-1}g_i(\alphabf^{(i)})$ in \eqref{eq:regret} from below, consider a full-information adversarial online linear maximization problem~\cite{hazan2016introduction, shalev2012online} for rounds $i=1,2,\ldots,\tau-1$~\footnote{Note that $\tau$ is a random variable, and hence the number of rounds in this online linear optimization problem is stochastic.}, where at each round the decision maker (player 1) chooses the dual vector $\alphabf^{(i)}\in\{\alphabf\in[0,1]^m:\lVert\alphabf\rVert_1\leq 1\}$, and an adaptive adversary (player 2) chooses the linear cost function $g_i(\alphabf)=\alphabf\cdot({\xbf}_i-\frac{1}{m})$ defined above. For any given adversarial realization of random variables $\{{\xbf}_i\}$, which defines the strategies of player 2, the goal of player 1 is to produce a sequence $\alphabf^{(1)},\alphabf^{(2)},\ldots,\alphabf^{(\tau-1)}$ that maximizes the linear objective function $\sum_{i=1}^{\tau-1}g_i(\alphabf^{(i)})$. 

Now, consider ``dual update'' steps of Algorithm~\ref{alg:onlineLMB}. These steps are essentially equivalent to player 1 running the exponential weight forecaster (also known as multiplicative weight updates) as a vanishing regret online learning algorithm in the aforementioned online linear maximization problem. This algorithm, which is parametric with learning rate $\eta$, picks the sequence $\alphabf^{(i)}$ for $i=1,2,\ldots,\tau-1$ to be the exponential weights distributions with weights $\eta\cdot k_j$, and guarantees the following regret bound~\cite{hazan2016introduction}: 
\begin{equation}
\label{eq:secondbound}
\sum_{i=1}^{\tau-1}g_i(\alphabf^{(i)})\geq (1-\eta)\underset{\lVert\alphabf\rVert_1\leq 1,\alphabf\geq \mathbf{0}}{\max}\sum_{i=1}^{\tau-1}g_i(\alphabf)-\frac{\log m}{\eta}\geq (1-\eta)(k-\frac{\tau-1}{m})-\frac{\log m}{\eta}
\end{equation}
where the last inequality holds because at the time $\tau-1$, either there exists $j$ such that $\sum_{i=1}^{\tau-1}{x}_{i,j}= k$, or $\tau-1=mk$ and all surrogate outcome budgets are exhausted. 
In the former case, we have  
$$\underset{\lVert\alphabf\rVert_1\leq 1,\alphabf\geq \mathbf{0}}{\max}\sum_{i=1}^{\tau-1}g_i(\alphabf)\geq \sum_{i=1}^{\tau-1}g_i(\mathbf{e}_j)\geq k-\frac{\tau-1}{m},$$ 
and in the latter case we have
$$\underset{\lVert\alphabf\rVert_1\leq 1,\alphabf\geq \mathbf{0}}{\max}\sum_{i=1}^{\tau-1}g_i(\alphabf)\geq 0\geq k-\frac{\tau-1}{m}.$$
By applying the regret bound in \eqref{eq:secondbound} to the RHS of the inequality in \eqref{eq:regret}, we have:
\begin{align*}
\sum_{i=1}^{\tau-1}\expect{\overline{\alg{}}_i \given \hist{i-1}}&\geq \frac{\tau-1}{km}\opt+\gamma(1-\eta)(k-\frac{\tau-1}{m})-\gamma\frac{\log m}{\eta}-\sum_{i=1}^{\tau-1}\lvert O_i\rvert -\sum_{i=1}^{\tau-1}\lvert Z_i \rvert+\sum_{i=1}^{\tau-1}L_i\\
&\overset{(1)}{\geq} \frac{\tau-1}{km}\opt+\opt(1-\eta)(1-\frac{\tau-1}{km})-\gamma\frac{\log m}{\eta}-\sum_{i=1}^{\tau-1}\lvert O_i\rvert -\sum_{i=1}^{\tau-1}\lvert Z_i \rvert+\sum_{i=1}^{\tau-1}L_i\\
&\overset{(2)}{\geq}\opt\left(1-\eta-O(1)\frac{\log m}{k\eta}\right)-\sum_{i=1}^{\tau-1}\lvert O_i\rvert -\sum_{i=1}^{\tau-1}\lvert Z_i \rvert+\sum_{i=1}^{\tau-1}L_i\\
&\overset{(3)}{\geq}\opt -\eta\left(km+km~O\left(\frac{\log m}{m\log(m/\eta)}\right)\right)-\sum_{i=1}^{\tau-1}\lvert O_i\rvert-\sum_{i=1}^{\tau-1}\lvert Z_i \rvert+\sum_{i=1}^{\tau-1}L_i\\
&=\opt-O(\eta km)-\sum_{i=1}^{\tau-1}\lvert O_i\rvert-\sum_{i=1}^{\tau-1}\lvert Z_i \rvert+\sum_{i=1}^{\tau-1}L_i\
\end{align*} 
where inequality~(1) above holds as $\gamma\geq \opt/k$, inequality~(2) above holds as $\gamma\leq O(1)~\opt/k$, and inequality~(3) above holds because $k\geq m\log(m/\eta)/\eta^2$ and $\opt(\vbf)\leq km$. Further notice that for times $i=1,2,\ldots,\tau-1$, if we consider history paths $\hist{i-1}$ in which $i\leq \tau-1$, we have:
\begin{align*}
\expect{{\alg{i}}\given \hist{i-1},\pi(i)}&=\sum_{j\in [m]}\expect{{x}_{i,j}\given \hist{i-1},\pi(i)} v_{\pi(i),j}\\
&=\sum_{j\in[m]}\expect{\overline{x}_{i,j}\given \hist{i-1},\pi(i)}v_{\pi(i),j}\geq \expect{\overline{\alg{}}_i\given \hist{i-1},\pi(i)} ~.
\end{align*}
Therefore, by taking expectation with respect to $\pi(i)$, we have $\expect{{\alg{}}_i\given \hist{i-1}}\geq \expect{\overline{\alg{}}_i\given \hist{i-1}}$ for $i=1,2,\ldots,\tau-1$. Denoting $\alg{}=\sum_{i=1}^{km}\alg{i}$, we have: 
\begin{align}
\expect{\alg{}}&\geq \expect{\sum_{i=1}^{\tau-1}\alg{i}}= \expect{\sum_{i=1}^{\tau-1}\expect{\alg{i}\given\hist{i-1}}}\geq \expect{\sum_{i=1}^{\tau-1}\expect{\overline{\alg{}}_{i}\given\hist{i-1}}} \nonumber\\
&\geq \opt-O(\eta km)-\expect{\sum_{i=1}^{\tau-1}\lvert O_i\rvert} -\expect{\sum_{i=1}^{\tau-1}\lvert Z_i \rvert}+\expect{\sum_{i=1}^{\tau-1}L_i}\nonumber\\
\label{eq:last-RHS}
&\geq  \opt-O(\eta km)-\expect{\sum_{i=1}^{km}\lvert O_i\rvert} -\expect{\sum_{i=1}^{km}\lvert Z_i \rvert}+\expect{\sum_{i=1}^{\tau-1}L_i}~.
\end{align}
 
We now finish the proof by:
\begin{enumerate}
\item[(i)] Showing  $\expect{\sum_{i=1}^{\tau-1}L_i}=0$, 
\item[(ii)] Upper bounding $\expect{\sum_{i=1}^{km}\lvert O_i\rvert}$ by $O\left(\sqrt{km\log(km)}\right)$ and $\expect{\sum_{i=1}^{km}\lvert Z_i\rvert}$ by $O\left(\gamma\sqrt{km\log(km)}\right)$. 
\end{enumerate}
Assuming (i) and (ii) are done, and recalling that  $k\geq m\log(m/\eta)/\eta^2$, we have:
\begin{align}
\label{eq:final1}
\expect{\sum_{i=1}^{km}\lvert O_i\rvert}\leq O(\sqrt{km\log(km)})=km~O\left(\eta\sqrt{\frac{\log(m^2\log(m/\eta))-2\log(1/\eta)}{m^2\log(m/\eta)}}\right)=O(\eta k)~,
\end{align}
Similarly, recalling that $\gamma\leq O(1)~\opt(\mathbf{v})/k\leq O(m)$, we have: 
\begin{align}
\label{eq:final2}
\expect{\sum_{i=1}^{km}\lvert Z_i\rvert}\leq O(\gamma\sqrt{km\log(km)})=O(\eta km)~
\end{align}
Combining the bounds \eqref{eq:final1} and \eqref{eq:final2} with $\expect{\sum_{i=1}^{\tau-1}L_i}=0$ and plugging them in \eqref{eq:last-RHS} gives us the final result, i.e., $\expect{\alg{}}/{km}\geq {\opt}/{km}-O(\eta)$.\\

In order to show (i), we simply use the fact that conditioned on any history path $\hist{i-1}$ for which $\tau-1\geq i$, we have: 
$$
\expect{L_i\given\hist{i-1}}=\expect{\alphabf^{(i)}\cdot(\expect{\xbf_i\given\hist{i-1}}-\xbf_i)\given\hist{i-1}}=\alphabf^{(i)}\cdot (\expect{\xbf_i\given\hist{i-1}}-\expect{\xbf_i\given\hist{i-1}})=0~.
$$
Therefore, $\expect{\sum_{i=1}^{\tau-1}L_i}=\sum_{i=1}^{km}\prob{\tau-1\geq i}\ex{\hist{i-1}:\tau-1\geq i}{\expect{L_i\given \hist{i-1}}}=0.$\\

In order to show (ii), we use the technical Lemma~\ref{lemma:technical}:
\begin{itemize}
\item  To upper bound $\expect{\sum_{i=1}^{km}\lvert O_i\rvert}$, let $d=1$, $T=km$, $z_i=\opt_{\pi^{-1}(i)}$, and $y_i=\expect{\opt_{i+1}\given\hist{i}}$ for $i=1,2,\ldots,km$. Notice that
\begin{align*}
y_i&=\expect{\opt_{i+1}\given\hist{i}}=\frac{\sum_{i'=i+1}^{km}{\opt_{i'}}}{km-i}=\frac{\sum_{i'=i+1}^{km}z_{i'}}{km-i}~,\\
y_0&=\expect{\opt_1}=\expect{\opt_{i+1}}=\opt/km~,~\\
\lvert O_i\rvert &=\lvert y_i-y_0\rvert
\end{align*}
Thus, the conditions of Lemma~\ref{lemma:technical} are satisfied and $\expect{\sum_{i=1}^{km}\lvert O_i\rvert}\leq O(\sqrt{km\log(km)})$.
\item To upper bound $\expect{\sum_{i=1}^{km}\lvert Z_i\rvert}$, let $d=m$, $T=km$, $\zbf_i=\xbf^*_{\pi^{-1}(i)}$ and $\ybf_i=\expect{\xbf^*_{i+1}\given\hist{i}}$ for $i=1,2,\ldots,km$. Again, notice that
\begin{align*}
\ybf_i&=\expect{\xbf^*_{i+1}\given\hist{i}}=\frac{\sum_{i'=i+1}^{km}{\xbf^*_{i'}}}{km-i}=\frac{\sum_{i'=i+1}^{km}\zbf_{i'}}{km-i}~,\\
\ybf_0&=\expect{\xbf^*_1}=\expect{\xbf^*_{i+1}}=\sum_{i=1}^{km}\xbf^*_i/km~.
\end{align*}
Therefore the conditions in the statement of Lemma~\ref{lemma:technical} are satisfied. Moreover, note that 
$$
\lvert Z_i\rvert=\gamma\lvert\alphabf^{(i)}\cdot(\ybf_i-\ybf_0)\rvert\leq \gamma \underset{j\in[m]}{\max}\left\lvert y_{i,j}-y_{0,j}\right\rvert~,
$$
where the inequality above holds as $\lVert\alphabf^{(i)} \rVert_1\leq 1$. Therefore, 
$$\expect{\sum_{i=1}^{km}\lvert Z_i\rvert}\leq O(\gamma\sqrt{km(\log(km)+\log(m))})=O(\gamma\sqrt{km\log(km)})~,$$ which concludes the proof.
\end{itemize}

\endproof

\subsubsection*{\underline{Step \rom{2}}: parameter $\gamma$ and approximating the offline optimal.} Pre-setting $\gamma$ to be an estimate of the optimal objective of the convex program in Definition~\ref{defn:LBM} is necessary for guarantee of Algorithm~\ref{alg:onlineLMB} in Proposition~\ref{prop:compt-analysis}. Also, $\gamma$ should be set in a symmetric and incentive compatible way across replicas, to preserve stationarity property. To this end, we look at an instance generated by an independent random draw of $mk$ replicas (while fixing the surrogates). In such an instance, we estimate the expected values by sampling and taking the empirical mean for each edge in the replica-surrogate bipartite graph. We then solve the convex program exactly (which can be done in polytime using an efficient separation oracle). Obviously, this scheme is incentive compatible as we do not even use the reported type of true agent in our calculation for $\gamma$, and it is symmetric across replicas. We now show how this approach leads to a constant approximation to the optimal value of the offline program in~\ref{defn:LBM} with high probability. Note that this scheme should be run offline as a \emph{pre-processing} step before running our online Algorithm~\ref{alg:onlineLMB} used in Proposition~\ref{prop:compt-analysis}.

\begin{proposition}
\label{prop:gamma-est}
If $k=\Omega(\frac{\log(1/\epsilon')}{\delta^2 m(\log m)^2})$ for some $\epsilon'>0$, then there exist a polynomial time algorithm that approximately calculates $\opt(\vbf)/k$, that is, it outputs $\gamma$ such that
$$ \OPT(\vbf)/k\leq \gamma\leq O(1) \OPT(\vbf)/k $$
with probability at least $1-\epsilon'$. Moreover, this algorithm only needs polynomial in $m$, $k$, $1/\delta$ and $1/\epsilon'$ samples to black-box allocation $\alloc$.
\end{proposition}

To formalize the approximation scheme and to prove Proposition~\ref{prop:gamma-est}, first fix the surrogate type profile $\sbf$. For a given replica profile $\rbf$ and replica-surrogate edge $(i,j)$, 
let $v_{i,j}(r_i)= \ex{}{v(r_i,\alloc(s_j))}$ and $\hat{v}_{i,j}(r_i)$ be the empirical mean of $N$ samples of the random variable $v(r_i,\alloc(s_j))$. Suppose  $\vbf(\rbf)$ and $\hat{\vbf}(\rbf)$ be the corresponding vectors of expected values and empirical means under replica profile $\rbf$. Now, draw $\rbf'$  independently at random from the distribution of $\rbf$. We now show that $\OPT(\hat{\vbf}(\rbf'))$ is a constant approximation to $\OPT(\vbf(\rbf))$ with high probability, and therefore we can use $\OPT(\hat{\vbf}(\rbf'))$ to set $\gamma$. 

We prove this in two steps. In Lemma~\ref{lem:step-1-est} we show for a given $\rbf'$, $\OPT(\hat{\vbf}(\rbf'))$ is a constant approximation to $\OPT({\vbf}(\rbf'))$ with high probability over randomness in $\{\alloc(s_j)\}$. Then, in Lemma~\ref{lem:step-2-est} we show if $\rbf$ and $\rbf'$ are two random independent draws from replica profile distribution then $\OPT(\vbf(\rbf'))$ is a constant approximation to $\OPT(\vbf(\rbf))$ with high probability over randomness in $\rbf$ and $\rbf'$. These two pieces together prove our claim. 

\begin{lemma}
\label{lem:step-1-est}
If $N\geq \frac{\log(4m^2 k/\epsilon')}{\delta^{2}(\log m) ^{2}}$, then 
$1/2\cdot\OPT(\vbf(\rbf'))\leq\OPT(\hat{\vbf}(\rbf'))\leq 2\OPT(\vbf(\rbf'))$
with probability at least $1-\epsilon'/2$.
\end{lemma}
\proof
By using standard Chernoff-Hoeffding bound together with union bound, with probability at least $1-2m^2ke^{-\frac{\delta^2(\log m)^2\cdot N}{2}}\geq 1-\epsilon'/2$ we have 
$$\forall (i,j)\in[km]\times[m]:~~~ \lvert \hat{v}_{i,j}(r'_i)-v_{i,j}(r'_i)\rvert \leq 1/2\cdot\delta\log m$$ 
Suppose $\xbf^{*}$ be the optimal solution of the regularized matching convex program with values $\vbf(\rbf')$ and  $\xbf^{**}$ be the optimal solution with values $\hat{\vbf}(\rbf')$. Denoting the entropy function by $H(\cdot)$, we have:
\begin{align}
\OPT(\hat{\vbf}(\rbf'))=&  \sum_{i}\left(\xbf^{**}_i \cdot \hat\vbf_i+\delta H(\xbf^{**}_i)\right)\geq  \sum_{i}\left(\xbf^{*}_i \cdot \hat\vbf_i+\delta H(\xbf^{*}_i)\right)\nonumber\\
\geq&  \sum_{i}\left(\xbf^{*}_i \cdot \vbf_i+\delta H(\xbf^{*}_i)\right)-\frac{\delta km\log m}{2}=\OPT(\vbf(\rbf'))-\frac{\delta km\log m}{2}\geq 1/2\cdot\OPT(\vbf(\rbf'))\nonumber
\end{align}
where the last inequality holds as $\OPT(\vbf(\rbf'))$ is bounded below by the value of the uniform allocation, i.e. $\OPT(\vbf(\rbf'))\geq \delta\cdot mk\log(m)$. Similarly, one can show $\OPT(\vbf(\rbf'))\geq 1/2\cdot\OPT(\hat\vbf(\rbf'))$, which completes the proof.
\endproof

Before proving the second step, we prove the following lemma, which basically shows the optimal value of regularized matching $\OPT(\vbf(\cdot))$ is a 1-Lipschitz multivariate function.
\begin{lemma}
\label{lem:1-lib}
For every $i\in[km]$, replica profile $\rbf$, and replica type $r'_i$ we have:
$$ \lvert\OPT(\vbf(r_i,r_{-i}))-\OPT(\vbf(r'_i,r_{-i}))\rvert \leq 1$$
\end{lemma}
\begin{proof}
Let $\xbf$ and $\xbf'$ be the optimal assignments in $\OPT(\vbf(r_i,r_{-i}))$ and $\OPT(\vbf(r'_i,r_{-i}))$ respectively. We have
\begin{align*}
\OPT(\vbf(r_i,r_{-i}))&=\sum_{l}(\xbf_{l}\cdot \vbf_{l}(r_l)+H(\xbf_{l}))\geq \sum_{l\neq i}(\xbf'_{l}\cdot \vbf_{l}(r_l)+H(\xbf'_{l}))+\xbf'_{i}\cdot \vbf_{i}(r_i)+H(\xbf'_{i})\\
&\geq \sum_{l\neq i}(\xbf'_{l}\cdot \vbf_{l}(r_l)+H(\xbf'_{l}))+\xbf'_{i}\cdot \vbf_{i}(r'_i)+H(\xbf'_{i})-1=\OPT(\vbf(r'_i,r_{-i}))-1
\end{align*}
where the last inequality holds because $\xbf'_i\cdot(\vbf_i(r'_i)-\vbf_i(r_i))\leq 1$. Similarly, $\OPT(\vbf(r'_i,r_{-i}))\geq \OPT(\vbf(r_i,r_{-i}))-1$ by switching the roles of $r_i$ and $r'_i$.
\end{proof}

\begin{lemma} If $k\geq \frac{32 \log(8/\epsilon')}{\delta^2 m(\log m)^2}$, then $ 1/2\cdot\OPT(\vbf(\rbf))\leq\OPT(\vbf(\rbf'))\leq 3/2\cdot \OPT(\vbf(\rbf))$  with probability at least $1-\epsilon'/2$.   
\label{lem:step-2-est}
\end{lemma}
\proof
The lemma can be proved directly by McDiarmid’s inequality~\cite{mcdiarmid1989method} and using the 1-Lipschitzness proved in Lemma~\ref{lem:1-lib}. For completeness, we prove it here from first principles. We start by defining the following Doob martingale sequence~\citep{MR-10}, where (conditional) expectations are taken over the randomness in the replica profile $\rbf$ :
\begin{align*}
X_0&=\ex{}{\OPT(\vbf(\rbf))}\\
X_n&=\ex{}{\OPT(\vbf(\rbf))|r_1,\cdots,r_n},~~~n=1,2,\ldots,km
\end{align*}
It is easy to check that $\ex{}{X_n|r_1,\ldots,r_{n-1}}=X_{n-1}$, and therefore $\{X_n\}$ forms a martingale sequence with respect to $\{r_n\}$. Moreover, $\lvert X_n -X_{n-1}\rvert \leq 1$ because of Lemma~\ref{lem:1-lib}. Now, by using Azuma–Hoeffding bound for martingales, we have
$$\Pr\{\lvert X_{km}-X_0 \rvert \geq \frac{\delta km \log(m)}{4}\}\leq 2e^{\frac{-km\delta^2(\log m)^2}{32}} $$
 and therefore w.p. at least $1-2e^{\frac{-km\delta^2(\log m)^2}{32}}$, we have $\lvert \OPT(\vbf(\rbf))-\ex{}{\OPT(\vbf(\rbf))}\rvert\leq \frac{\delta km \log(m)}{4}$. Similarly, w.p. at least $1-2e^{\frac{-km\delta^2(\log m)^2}{32}}$, we have $\lvert \OPT(\vbf(\rbf'))-\ex{}{\OPT(\vbf(\rbf))}\rvert\leq \frac{\delta km \log(m)}{4}$. Therefore w.p. at least $1-4e^{\frac{-km\delta^2(\log m)^2}{32}}$ we have $\lvert \OPT(\vbf(\rbf))-\OPT(\vbf(\rbf'))\rvert\leq \frac{\delta km \log(m)}{2}$. By using the lower-bound of $\delta km \log (m)$ for $\OPT(\vbf(\rbf))$ (due to uniform assignment), we conclude that with probability at least $1-4e^{\frac{-km\delta^2(\log m)^2}{32}}\geq 1-\epsilon'/2$ we have the following, as desired:
\begin{equation*}
 1/2\cdot\OPT(\vbf(\rbf))\leq\OPT(\vbf(\rbf'))\leq 3/2\cdot \OPT(\vbf(\rbf)) \qedhere
  \end{equation*}  
\endproof

By putting Lemma~\ref{lem:step-1-est} and Lemma~\ref{lem:step-2-est} together, we immediately get the following corollary.
\begin{corollary} 
\label{cor:estimating-offline}
If $N\geq \frac{\log(4m^2 k/\epsilon')}{\delta^{2}(\log m) ^{2}}$ and $k\geq \frac{32 \log(8/\epsilon')}{\delta^2 m(\log m)^2}$, then $\gamma=\frac{4}{k}\OPT(\hat\vbf(\rbf'))$ satisfies 
$$\OPT(\vbf(\rbf))/k \leq \gamma \leq 12\cdot \OPT(\vbf(\rbf))/k~,$$
with probability at least $1-\eta$.
\end{corollary}

We conclude the above discussion as the proof of Proposition~\ref{prop:gamma-est} is immediate from Corollary~\ref{cor:estimating-offline}.

\subsubsection*{\underline{Step 3:} convex objective of optimal entropy regularized matching vs. welfare of optimal matching. }
The last ingredient we need in the proof is the following lemma.
\begin{lemma} 
\label{lem:reg-matching-welfare}
With parameter $\delta \geq 0$ the per-replica convex objective value of the optimal entropy regularized matching is within an
additive $\delta \log m$ of the welfare of the optimal matching.
\end{lemma}
 
\proof
The entropy $-\sum_{i,j} x_{i,j} \log x_{i,j}$ is non-negative and
maximized with $x_{i,j} = 1/m$.  The maximum value of the entropy term
is thus $\delta mk \log m$.  The optimal convex objective value of the entropy regularized matching exceeds that of the optimal matching;
thus, it is within an additive $\delta mk \log m$ of the welfare of the optimal (unregularized) matching. As a result, the per-replica convex objective value of the optimal entropy regularized matching is within $\delta \log m$ of the per-replica welfare of the optimal matching.

\endproof

\subsubsection*{\underline{Step 4:} putting all the pieces together} We conclude the section by combining Propositions~\ref{prop:gamma-est},~\ref{prop:compt-analysis}, and Lemma~\ref{lem:reg-matching-welfare} to prove the main theorem.

\proof[Proof of Theorem~\ref{thm:main-bic-reduction}]
Let $\delta=\frac{\epsilon}{3}\cdot\frac{1}{\log m}$ and $\eta=\frac{\epsilon}{3}\cdot \frac{1}{c}$, where $c$ is a constant such that the per-replica welfare of Algorithm~\ref{alg:onlineLMB} is within an additive $c\cdot \eta$ of the per-replica offline optimal objective value of the entropy regularized matching problem when estimation $\gamma$ is set appropriately (Proposition~\ref{prop:compt-analysis}). Moreover, let $k=\frac{m\log (m/\eta)}{\eta^2}=\Theta(\frac{m\log (m/\epsilon)}{\epsilon^2})$, to satisfy the required condition in Proposition~\ref{prop:compt-analysis}. Finally, set $\epsilon'=\frac{\epsilon}{3}$ and  the number of samples $N$ in the scheme described in Corollary~\ref{cor:estimating-offline} to $N=\Theta(\frac{\log(m^2 k/\epsilon')}{\delta^{2}(\log m) ^{2}})=\Theta(\frac{\log(m/\epsilon)}{\epsilon^2(\log m)^2})$, so that $\gamma$ is the appropriate estimator used in Proposition~\ref{prop:compt-analysis} with probability $1-\frac{\epsilon}{3}$.

The expected per-replica welfare of Algorithm~\ref{alg:onlineLMB} is within an additive $c\eta=\epsilon/3$ of the per-replica optimal convex objective value of the entropy regularized matching, with probability at least $1-\epsilon$, due to Propositions~\ref{prop:compt-analysis} and~\ref{prop:gamma-est}. This probability is over the internal randomness of the samples used in the estimation scheme for $\gamma$ in Corollary~\ref{cor:estimating-offline}. The per-replica optimal objective value of the entropy regularized matching is at most $1$. Therefore, the overall expected per-replica welfare of Algorithm~\ref{alg:onlineLMB} is within an additive $\epsilon/3+\epsilon'=2\epsilon/3$ of the per-replica optimal convex objective value of the entropy regularized matching. Following Lemma~\ref{lem:reg-matching-welfare}, the per-replica optimal value of the entropy regularized matching is within an additive $\delta \log m=\epsilon/3$ of the per-replica welfare of the optimal matching, and therefore the expected per-replica welfare of Algorithm~\ref{alg:onlineLMB} is within an additive $2\epsilon/3+\epsilon/3=\epsilon$ of the per-replica welfare of the optimal matching.

Finally, due to Proposition~\ref{prop:gamma-est} and the fact that $k$ is polynomial in $m$ and $1/\epsilon$, the algorithm's running time is polynomial in $m$ and $1/\epsilon$. 
\endproof

\subsection{The End-to-End BIC Black-box Reduction}
We now summarize the proposed BIC black-box reduction. We incorporate our surrogate selection rule (By using Algorithm~\ref{alg:onlineLMB} as the matching algorithm in Definition~\ref{defn:rsba}) in the reduction under ideal-model proposed in \citet{HKM-15} and we set the market size parameter $m$ accordingly to maintain the  welfare preservation property of this reduction.

\begin{definition}[\citet{HKM-15}] The doubling dimension of a metric space is the smallest constant $\Delta$ such that every bounded subset $S$ can be partitioned into at most $2^\Delta$  subsets, each having diameter at most half of the diameter of $S$.
\end{definition}

We now use the following theorem in~\cite{HKM-15}, which states the welfare preservation of the maximum weight replica-surrogate matching in the ideal model if $m$ is large enough. 

\begin{theorem}[\citet{HKM-15}] 
\label{thm:HKM-doubling-dimension}
For any agent with type space $\typeS{}$ that has doubling dimension $\Delta\geq 2$, if 
$$m\geq \frac{1}{2\epsilon^{\Delta+1}}~,$$
then the expected per-replica welfare of the
maximum matching in the ideal model of \citet{HKM-15}  with load $k=1$ is within an additive $2\epsilon$ of the expected welfare of allocation $\alloc$ for that agent. 
\end{theorem}

We now have the following immediate corollary by combining Theorem~\ref{thm:main-bic-reduction} with Theorem~\ref{thm:HKM-doubling-dimension}.

\begin{corollary}[\emph{BIC black-box reduction}]
If the market size parameter $m$ is set to $\lceil\frac{1}{2\epsilon^{\Delta+1}}\rceil$, and the parameters of Algorithm~\ref{alg:onlineLMB} are set as stated in Theorem~\ref{thm:main-bic-reduction}, then the composition of surrogate selection rule defined by Algorithm~\ref{alg:onlineLMB} with the allocation $\alloc$ is (1)~a BIC mechanism, (2)~the expected welfare is within an additive $3\epsilon$ of the expected welfare of $\alloc$ for each agent, and (3)~its running time is polynomial in $n$ and $1/\epsilon$ given access to black-box oracle $\alloc$.\footnote{Our result obviously holds when the doubling dimensions of type spaces are considered to be constant. For arbitrary large-dimensional type spaces, the running time is polynomial in $n$ and $1/\epsilon^{\Delta}$.}
\end{corollary}

\section{Implicit Payment Computations}
\label{sec:implicitpayments}
In this section we describe one standard reduction for computing implicit payments in our general setting, given access to a BIC allocation algorithm $\tilde{\alloc}$: a multi-parameter
counterpart of the single-parameter payment computation
procedure used for example by \citet{APTT-04,HL-10}, 
which makes $n+1$ calls to $\tilde{\alloc}$, thus incurring a factor
$n+1$ overhead in running time. A different implicit payment
computation procedure, described in~\citet{BKS-13,BKS-15}, 
avoids this overhead by calling $\tilde{\alloc}$ only once
in expectation, but incurs a $1-\epsilon$ loss in expected welfare
and potentially makes payments of magnitude $\Theta(1/\epsilon)$ 
{\em from} the mechanism {\em to} the agents.

The implicit payment computation procedure assumes that
the agents' type spaces $(\typeS{k})_{k \in [n]}$ are
{\em star-convex at 0}, meaning that 
for any agent $k$, any type $\type{k} \in \typeS{k}$, 
and any scalar $\lambda \in [0,1]$, there is another
type $\lambda \type{k} \in \typeS{k}$ with the property
that $v(\lambda \type{k}, o) = \lambda v(\type{k},o)$
for every $o \in \outS$. (The assumption is without
loss of generality, as argued in the next paragraph.)
The implicit payment computation procedure, applied
to type profile $\typeP$, samples 
$\lambda \in [0,1]$ uniformly at random and computes
outcomes $o^0 \triangleq \tilde{\alloc}(\typeP)$ as well as
$o^k \triangleq \tilde{\alloc}(\lambda \type{k},\typeO{-k})$
for all $k \in [n]$. The payment charged to agent $k$
is $v(\type{k},o^0)-v(\type{k},o^k)$. Note that, in
expectation, agent $k$ pays
\[
  \pay{k}(\typeP) = v(\type{k},\tilde{\alloc}(\typeP)) -
    \int_{0}^{1} 
    v(\type{k},\tilde{\alloc}(\lambda \type{k},\typeO{-k}))
    \, d\lambda,
\]
in accordance with the payment identity for 
multi-parameter BIC mechanisms when type
spaces are star-convex at 0; see~\citet{BKS-13}
for a discussion of this payment identity.

Finally, let us justify the assumption that
$\typeS{k}$ is star-convex for all $k$. 
This assumption is without
loss of generality for the allocation algorithms
$\tilde{\alloc}$ that arise from the RSM reduction, because
we can enlarge the type space $\typeS{k}$ if necessary
by adjoining types of the form $\lambda \type{k}$
with $\type{k} \in \typeS{k}$ and $0 \leq \lambda < 1$.
Although the output of the original allocation algorithm
$\alloc$ may be undefined when its input type profile
includes one of these artifically-adjoined types,
the RSM reduction never inputs such a type into $\alloc$.
It only calls $\alloc$ on profiles
of surrogate types sampled from
the type-profile distribution $\prior{}$, whose support
excludes the artificially-adjoined types.
Thus, even when the input to $\tilde{\alloc}$ includes
an artifically-adjoined type $\lambda \type{k}$, it occurs
as one of the replicas in the reduction. The behavior of 
algorithm $\tilde{\alloc}$ remains well-defined in this 
case, because replicas are
only used as inputs to the valuation function $v(r_i,o_j)$,
whose output is well-defined even when $r_i = \lambda \type{k}$
for $\lambda < 1$.

\section{Summary and Future Directions}
In this paper we investigated the question of designing Bayesian incentive compatible blackbox reductions in mechanism design. We provided a polynomial time reduction from Bayesian incentive compatible mechanism design to Bayesian algorithm design for welfare maximization problems.Unlike prior results, our reduction achieves exact incentive compatibility for problems with multi-dimensional and continuous type spaces. We showed how to  employ and generalize the computational model in the literature on Bernoulli Factories. In particular we considered a generalization which we called the expectations from samples computational model, in which a problem instance is specified by a function mapping the expected values of a set of input distributions to a distribution over outcomes. The challenge is to give a polynomial time algorithm that exactly samples from the distribution over outcomes given only sample access to the input distributions. In this model, we gave a polynomial time algorithm for the function given by exponential weights: expected values of the input distributions correspond to the weights of alternatives and we wish to select an alternative with probability proportional to an exponential function of its weight.  As we showed, this algorithm is the key ingredient in designing an incentive compatible mechanism for bipartite matching, which can be used to make the approximately incentive compatible reduction of \citet{HKM-15} exactly incentive compatible. 

While the existence of such a reduction is good news for Bayesian
mechanism design, there are limitations that are mostly unavoidable.
\begin{itemize}

\item{\emph{Beyond expected social welfare.}}
It is tempting to try converting an arbitrary
algorithm  for  an  optimization  problem  into  a  computionally  efficient  Basyesian
truthful mechanism.  Interestingly,  this is not possible for all optimization objectives. In particular, \citet{CIL-12} show that no black-box reduction is possible
for the objective of makespan, even if we only require Bayesian truthfulness and an
average-case performance guarantee.  This precludes extending our result beyond
the expected-welfare objective in a general fashion.
\item{\emph{Exponential  dependence  on  dimension.}}
Notably,  our reduction is a fully polynomial time approximation scheme to the reduction of
\citet{HKM-15}.   However,  the  running  time  of  \citet{HKM-15} has
exponential dependence on $\Delta$. Therefore, our reduction also suffers from
the  same  exponential  dependence.   Intuitively,  this  seems  to  be  unavoidable  for
reductions that can only access the type space by sampling and can only access the
outcome space by calling the allocation function on sampled type profiles.
\end{itemize}

We  conclude  with  some  open  questions.   The  first  natural  question,  directly
related  to  the  second  limitation  above,  is  to  determine  whether  or  not  the  exponential dependence on $\Delta$ in the black-box reduction is unavoidable.  Are there
black-box reductions whose running time exhibits a milder dependence on the structure of the type space?  Another interesting question is to find more connections between Bayesian mechanism design and the expectations from samples computational model. Finally, one might be interested in a generalization of Bernoulli race to more interesting combinatorial settings, e.g. can one sample a base of a matroid given access to marginal coins, so that the sampling procedure satisfies the marginals \emph{exactly}? It would also be interesting to see  if these tools result in simpler blackbox reductions. 

\section*{Acknowledgment}
The authors would like to thank Nikhil Devanur, Shipra Agrawal and Pooya Jalaly for helpful discussions and comments on various parts of the paper. The first author was supported by NSF CAREER Award CCF-1350900. The second author was supported in part by NSF grant  CCF 1618502. The third author was partially supported by NSF grant CCF-1512964. We should also want to emphasize that an early conference version of this work was presented at The Annual ACM Symposium on Theory of Computing (STOC)~\citep{DHKN-17}.
\bibliographystyle{plainnat}
\bibliography{refs}

\begin{thebibliography}{30}
\providecommand{\natexlab}[1]{#1}
\providecommand{\url}[1]{\texttt{#1}}
\expandafter\ifx\csname urlstyle\endcsname\relax
  \providecommand{\doi}[1]{doi: #1}\else
  \providecommand{\doi}{doi: \begingroup \urlstyle{rm}\Url}\fi

\bibitem[Agrawal and Devanur(2015)]{agrawal2015fast}
Shipra Agrawal and Nikhil~R Devanur.
\newblock Fast algorithms for online stochastic convex programming.
\newblock In \emph{Proceedings of the Twenty-Sixth Annual ACM-SIAM Symposium on
  Discrete Algorithms}, pages 1405--1424. SIAM, 2015.

\bibitem[Agrawal et~al.(2009)Agrawal, Wang, and Ye]{agrawal2009dynamic}
Shipra Agrawal, Zizhuo Wang, and Yinyu Ye.
\newblock A dynamic near-optimal algorithm for online linear programming.
\newblock \emph{arXiv preprint arXiv:0911.2974}, 2009.

\bibitem[Archer et~al.(2004)Archer, Papadimitriou, Talwar, and Tardos]{APTT-04}
Aaron Archer, Christos Papadimitriou, Kunal Talwar, and {\'E}va Tardos.
\newblock An approximate truthful mechanism for combinatorial auctions with
  single parameter agents.
\newblock \emph{Internet Mathematics}, 1\penalty0 (2):\penalty0 129--150, 2004.

\bibitem[Babaioff et~al.(2013)Babaioff, Kleinberg, and Slivkins]{BKS-13}
Moshe Babaioff, Robert Kleinberg, and Aleksandrs Slivkins.
\newblock Multi-parameter mechanisms with implicit payment computation.
\newblock In \emph{Proceedings of the 14th ACM Conference on Electronic
  Commerce}, pages 35--52, 2013.

\bibitem[Babaioff et~al.(2015)Babaioff, Kleinberg, and Slivkins]{BKS-15}
Moshe Babaioff, Robert Kleinberg, and Aleksandrs Slivkins.
\newblock Truthful mechanisms with implicit payment computation.
\newblock \emph{J. ACM}, 62\penalty0 (2):\penalty0 10:1--10:37, May 2015.
\newblock ISSN 0004-5411.

\bibitem[Badanidiyuru et~al.(2013)Badanidiyuru, Kleinberg, and
  Slivkins]{badanidiyuru2013bandits}
Ashwinkumar Badanidiyuru, Robert Kleinberg, and Aleksandrs Slivkins.
\newblock Bandits with knapsacks.
\newblock In \emph{Foundations of Computer Science (FOCS), 2013 IEEE 54th
  Annual Symposium on}, pages 207--216. IEEE, 2013.

\bibitem[Bei and Huang(2011)]{BH-11}
Xiaohui Bei and Zhiyi Huang.
\newblock Bayesian incentive compatibility via fractional assignments.
\newblock In \emph{Proceedings of the 22nd ACM-SIAM Symposium on Discrete
  Algorithms}, pages 720--733. SIAM, 2011.

\bibitem[Boyd et~al.(2004)Boyd, Boyd, and Vandenberghe]{boyd2004convex}
Stephen Boyd, Stephen~P Boyd, and Lieven Vandenberghe.
\newblock \emph{Convex optimization}.
\newblock Cambridge university press, 2004.

\bibitem[Bubeck(2014)]{bubeck2014convex}
S{\'e}bastien Bubeck.
\newblock Convex optimization: Algorithms and complexity.
\newblock \emph{arXiv preprint arXiv:1405.4980}, 2014.

\bibitem[Chawla et~al.(2012)Chawla, Immorlica, and Lucier]{CIL-12}
Shuchi Chawla, Nicole Immorlica, and Brendan Lucier.
\newblock On the limits of black-box reductions in mechanism design.
\newblock In \emph{Proceedings of the 44th ACM Symposium on Theory of
  Computing}, pages 435--448. ACM, 2012.

\bibitem[Chen and Wang(2015)]{chen2015dynamic}
Xiao~Alison Chen and Zizhuo Wang.
\newblock A dynamic learning algorithm for online matching problems with
  concave returns.
\newblock \emph{European Journal of Operational Research}, 247\penalty0
  (2):\penalty0 379--388, 2015.

\bibitem[Devanur et~al.(2011)Devanur, Jain, Sivan, and
  Wilkens]{devanur2011near}
Nikhil~R Devanur, Kamal Jain, Balasubramanian Sivan, and Christopher~A Wilkens.
\newblock Near optimal online algorithms and fast approximation algorithms for
  resource allocation problems.
\newblock In \emph{Proceedings of the 12th ACM conference on Electronic
  commerce}, pages 29--38. ACM, 2011.

\bibitem[Dughmi et~al.(2017)Dughmi, Hartline, Kleinberg, and Niazadeh]{DHKN-17}
Shaddin Dughmi, Jason~D Hartline, Robert Kleinberg, and Rad Niazadeh.
\newblock Bernoulli factories and black-box reductions in mechanism design.
\newblock In \emph{Proceedings of the 49th Annual ACM SIGACT Symposium on
  Theory of Computing}, pages 158--169, 2017.

\bibitem[Groves(1973)]{G-73}
Theodore Groves.
\newblock Incentives in teams.
\newblock \emph{Econometrica: Journal of the Econometric Society}, pages
  617--631, 1973.

\bibitem[Hartline and Lucier(2010)]{HL-10}
Jason~D Hartline and Brendan Lucier.
\newblock Bayesian algorithmic mechanism design.
\newblock In \emph{Proceedings of the forty-second ACM symposium on Theory of
  computing}, pages 301--310. ACM, 2010.

\bibitem[Hartline and Lucier(2015)]{HL-15}
Jason~D Hartline and Brendan Lucier.
\newblock Non-optimal mechanism design.
\newblock \emph{The American Economic Review}, 105\penalty0 (10):\penalty0
  3102--3124, 2015.

\bibitem[Hartline et~al.(2011)Hartline, Kleinberg, and Malekian]{HKM-11}
Jason~D. Hartline, Robert Kleinberg, and Azarakhsh Malekian.
\newblock Bayesian incentive compatibility via matchings.
\newblock \emph{SODA}, 2011.

\bibitem[Hartline et~al.(2015)Hartline, Kleinberg, and Malekian]{HKM-15}
Jason~D. Hartline, Robert Kleinberg, and Azarakhsh Malekian.
\newblock Bayesian incentive compatibility via matchings.
\newblock \emph{Games and Economic Behavior}, 92\penalty0 (C):\penalty0
  401--429, 2015.

\bibitem[Hazan(2016)]{hazan2016introduction}
Elad Hazan.
\newblock Introduction to online convex optimization.
\newblock \emph{Foundations and Trends in Optimization}, 2\penalty0
  (3-4):\penalty0 157--325, 2016.

\bibitem[Huang and Kannan(2012)]{HuangKannan-12}
Zhiyi Huang and Sampath Kannan.
\newblock The exponential mechanism for social welfare: {P}rivate, truthful,
  and nearly optimal.
\newblock In \emph{Proceedings of the 53rd IEEE Symposium on Foundations of
  Computer Science}, pages 140--149. IEEE, 2012.

\bibitem[Huber(2015)]{Huber2015OptimalLB}
Mark Huber.
\newblock Optimal linear bernoulli factories for small mean problems.
\newblock \emph{CoRR}, abs/1507.00843, 2015.

\bibitem[Kalai and Vempala(2005)]{kalai2005efficient}
Adam Kalai and Santosh Vempala.
\newblock Efficient algorithms for online decision problems.
\newblock \emph{Journal of Computer and System Sciences}, 71\penalty0
  (3):\penalty0 291--307, 2005.

\bibitem[Keane and O'Brien(1994)]{keane1994bernoulli}
MS~Keane and George~L O'Brien.
\newblock A bernoulli factory.
\newblock \emph{ACM Transactions on Modeling and Computer Simulation (TOMACS)},
  4\penalty0 (2):\penalty0 213--219, 1994.

\bibitem[Kesselheim et~al.(2014)Kesselheim, T{\"o}nnis, Radke, and
  V{\"o}cking]{kesselheim2014primal}
Thomas Kesselheim, Andreas T{\"o}nnis, Klaus Radke, and Berthold V{\"o}cking.
\newblock Primal beats dual on online packing lps in the random-order model.
\newblock In \emph{Proceedings of the 46th Annual ACM Symposium on Theory of
  Computing}, pages 303--312. ACM, 2014.

\bibitem[{\L}atuszy{\'n}ski(2010)]{latuszynski2010bernoulli}
Krzysztof {\L}atuszy{\'n}ski.
\newblock The bernoulli factory, its extensions and applications.
\newblock \emph{Proceedings of IWAP 2010}, pages 1--5, 2010.

\bibitem[McDiarmid(1989)]{mcdiarmid1989method}
Colin McDiarmid.
\newblock On the method of bounded differences.
\newblock \emph{Surveys in combinatorics}, 141\penalty0 (1):\penalty0 148--188,
  1989.

\bibitem[Motwani and Raghavan(2010)]{MR-10}
Rajeev Motwani and Prabhakar Raghavan.
\newblock \emph{Randomized algorithms}.
\newblock Chapman \& Hall/CRC, 2010.

\bibitem[Nacu and Peres(2005)]{nacu2005fast}
{\c{S}}erban Nacu and Yuval Peres.
\newblock Fast simulation of new coins from old.
\newblock \emph{The Annals of Applied Probability}, 15\penalty0 (1A):\penalty0
  93--115, 2005.

\bibitem[Nisan and Ronen(2001)]{NR-01}
Noam Nisan and Amir Ronen.
\newblock Algorithmic mechanism design.
\newblock \emph{Games and Economic behavior}, 35\penalty0 (1-2):\penalty0
  166--196, 2001.

\bibitem[Shalev-Shwartz et~al.(2012)]{shalev2012online}
Shai Shalev-Shwartz et~al.
\newblock Online learning and online convex optimization.
\newblock \emph{Foundations and Trends{\textregistered} in Machine Learning},
  4\penalty0 (2):\penalty0 107--194, 2012.

\end{thebibliography}
\appendix

\section{Surrogate Selection and BIC Reduction - Further Details}
\label{app:omitted}

\begin{lemma}
\label{lem:dispreserv}
If matching algorithm $M(\rbf,\sbf)$ produces a perfect $k$-to-$1$ matching for the instance in Definition~\ref{defn:rsba}, then its corresponding surrogate selection rule, denoted by $\SSR^{M}$, is stationary
\end{lemma}
\begin{proof}{Proof of Lemma~\ref{lem:dispreserv}.}
 Each surrogate $s_j$ is an i.i.d. sample from $\prior{}$. Moreover, by the principle of deferred decisions the index $i^*$ (the real agent's index in the replica type profile) is a uniform random index in $[mk]$, even after fixing the matching. Since this choice of replica is uniform in $[mk]$ and $M$ is a perfect $k$-to-$1$ matching, the selection of surrogate outcome is uniform in $[m]$, and therefore the selection of surrogate type associated with this outcome is also uniform in $[m]$. As a result, the output distribution of the selected surrogate type is $\prior{}$.
\end{proof}
\begin{lemma} 
\label{lem:singleBIC}
If $M(\rbf,\sbf)$ is a feasible replica-surrogate $k$-to-$1$ matching and is a truthful allocation rule (in expectation over allocation's random coins) for all replicas (i.e. assuming each replica is a rational agent, no replica has any incentive to misreport), then the composition of $\SSR^M$ and interim allocation algorithm $\alloc(.)$ forms a BIC allocation algorithm for the original mechanism design problem. 
\end{lemma}
\begin{proof}{Proof of Lemma~\ref{lem:singleBIC}.}
 Each replica-agent $i\in[mk]$ (\emph{including} the real agent $i^*$) bests off by reporting her true replica type under some proper payments. Now, consider an agent in the original mechanism design problem  with true type $t$. For any given surrogate type profile $\sbf$, using the $\SSR^M$-reduction the agent receives the same outcome distribution as the one he gets matched to in $M$ in a Bayesian sense, simply because of stationary property of $\SSR^M$ (Lemma~\ref{lem:dispreserv}). As allocation $M$ is incentive compatible, this agent doesn't benefit from miss-reporting her true type as long as the value he receives for reporting $t'$ is $v(t,\alloc(\SSR^M(t')))$.  Therefore conditioning on $\sbf$ and non-real replicas in $\rbf$, the final allocation is BIC from the perspective of this agent. The lemma then follows by averaging over the random choice of $\sbf$ and non-real agent replicas in $\rbf$.
\end{proof}

\end{document}